\documentclass[aps,prx,reprint,nofootinbib]{revtex4-2}
\pdfoutput=1

\usepackage{dcolumn}
\usepackage{physics,amsmath,amsthm,amsfonts,graphicx,times,xfrac,mathtools,mathrsfs,amssymb,bm,bbm,verbatim,appendix,mathrsfs,scalerel,lipsum,booktabs}
\usepackage[normalem]{ulem}
\usepackage{xcolor,hyperref}
\usepackage[caption=false,labelfont={normalsize}]{subfig}
\theoremstyle{plain}
\newtheorem{definition}{Definition}
\newtheorem{lemma}{Lemma}
\newtheorem{proposition}{Proposition}
\newtheorem{theorem}{Theorem}
\newtheorem{corollary}{Corollary}

\newcommand{\Def}{\coloneqq}
\newcommand{\Id}{\Tilde{\mathbb{I}}}
\newcommand{\Cbb}{\mathbb{C}}

\newcommand{\Lbb}{\mathbb{L}}
\newcommand{\Mbb}{\mathbb{M}}
\newcommand{\Sbb}{\mathbb{S}}
\newcommand{\Pbb}{\mathbb{P}}

\newcommand{\Ubb}{\mathbb{U}}
\newcommand{\Xbb}{\mathbb{X}}
\newcommand{\Ybb}{\mathbb{Y}}

\newcommand{\rmi}{\mathrm{i}}
\newcommand{\rmo}{\mathrm{o}}

\newcommand{\rmC}{\mathrm{C}}

\newcommand{\rmE}{\mathrm{E}}
\newcommand{\rmG}{\mathrm{G}}

\newcommand{\rmN}{\mathrm{N}}
\newcommand{\rmR}{\mathrm{R}}

\newcommand{\rmT}{\mathrm{T}}

\newcommand{\Acal}{\mathcal{A}}

\newcommand{\Ccal}{\mathcal{C}}

\newcommand{\CS}{\bm{\rmC}_{S}}
\newcommand{\CD}{\bm{\rmC}_{D}}

\newcommand{\TS}{\bm{\rmT}_{S}}
\newcommand{\NS}{\bm{\rmN}_{S}}
\newcommand{\RS}{\bm{\rmR}_{S}}
\newcommand{\Dcal}{\mathcal{D}}
\newcommand{\Ecal}{\mathcal{E}}
\newcommand{\ES}{\bm{\rmE}_{S}}
\newcommand{\Fcal}{\mathscr{F}}

\newcommand{\GS}{\bm{\rmG}_{S}}
\newcommand{\Hcal}{\mathcal{H}}
\newcommand{\Ical}{\mathcal{I}}
\newcommand{\Lcal}{\mathcal{L}}
\newcommand{\Mcal}{\mathcal{M}}
\newcommand{\Ncal}{\mathcal{N}}
\newcommand{\Pcal}{\mathcal{P}}
\newcommand{\Qcal}{\mathcal{Q}}
\newcommand{\Rcal}{\mathcal{R}}
\newcommand{\Scal}{\mathcal{S}}

\newcommand{\Ucal}{\mathcal{U}}

\newcommand{\Ycal}{\mathcal{Y}}
\newcommand{\Zcal}{\mathcal{Z}}

\begin{document}

\title{Quantifying and Bounding Spatiotemporal Correlations in Quantum Noise}

\author{Guilherme Zambon}
\email{guilhermezambon@usp.br}
\author{Diogo O. Soares-Pinto}
\email{dosp@ifsc.usp.br}
\affiliation{Sao Carlos Institute of Physics, University of Sao Paulo, IFSC – USP, 13566-590, Sao Carlos, SP, Brazil.}

\begin{abstract}
Spatial and temporal correlations in quantum noise challenge the local Markovian models commonly used in quantum information processing. We develop a unified operational framework for quantifying temporal, spatial, and total spatiotemporal correlations in general quantum processes. Using process tensors, we define these quantities through optimal distinguishability from Markovian, spatially local, and fully uncorrelated processes. Optimization over admissible probing combs ensures monotonicity under every superprocess that preserves the corresponding free set, while Choi-state functionals and restricted probes provide accessible lower bounds. We establish hierarchy and interpolation relations among the correlation measures and derive dimension-dependent upper bounds on temporal correlations transmitted by quantum memories, with tighter bounds for classical memory, together with universal ceilings imposed by the system dimension. These bounds turn certified lower estimates into witnesses of minimum memory dimension, nonclassicality under a memory-dimension constraint, and inconsistencies in the assumed process model. We illustrate the results with effective superconducting-qubit models featuring \(ZZ\) and \(XY\) interactions, exhibiting saturation of classical-memory, quantum-memory, and system-dimension ceilings. These results pave the way toward practical approaches to characterizing and addressing spatiotemporally correlated errors in quantum devices.

\end{abstract}

\maketitle

\section{Introduction}

Quantum computation offers a fundamentally new paradigm for information processing, with the potential to address problems that are beyond the practical reach of classical devices. Realizing this potential, however, requires the preparation, manipulation, and preservation of fragile quantum states over increasingly large and complex circuits. Any physical quantum processor inevitably interacts with uncontrolled degrees of freedom, and imperfections in state preparation, gates, measurements, and control lead to the progressive degradation of the encoded information. Noise therefore remains the central obstacle to extracting reliable computational results from noisy intermediate-scale quantum devices and, ultimately, to achieving large-scale fault-tolerant quantum computation \cite{nielsen2010quantum,preskill2018quantum,eisert2025mind}.

Common techniques for counteracting noise, such as quantum error correction and quantum error mitigation, often assume for simplicity and tractability that noise is uncorrelated across qubits and circuit layers \cite{terhal2015quantum,temme2017error,cai2023quantum}. In such descriptions, the noise is spatially local and Markovian: errors affecting different subsystems are independent, and the noise acting at one stage of the circuit carries no memory of previous stages. Although this approximation can provide an effective description in suitable regimes, it is not a generic property of physical quantum processors.

Realistic devices are subject to crosstalk, coherent residual couplings, common environmental modes, slowly fluctuating control parameters, leakage, and repeated interactions with microscopic defects \cite{sarovar2020detecting,klimov2018fluctuations,proctor2020detecting,mcewen2021removing}. These mechanisms naturally correlate errors across different subsystems and different times, as observed, for example, in superconducting-qubit processors \cite{lupke2020spectroscopy,white2023filtering,zou2024spatially}. Structured non-Markovian noise can be accommodated in fault-tolerance analyses under additional locality and strength assumptions \cite{terhal2005fault}, but general spatiotemporal correlations remain considerably harder to characterize and suppress than independent errors. This becomes particularly important as quantum processors approach the low physical and logical error rates required for fault-tolerant applications: approximations that are adequate for describing short noisy circuits may cease to be reliable when small correlations persist across many qubits, gates, and error-correction cycles \cite{fowler2014quantifying,chubb2021statistical,kam2025detrimental}. The distinction between these two descriptions is illustrated in Fig.~\ref{fig:circuits}: panel (a) shows the conventional local, memoryless noise model, whereas panel (b) represents the general spatiotemporally correlated noise considered here. Progress toward fault tolerance therefore requires moving beyond idealized uncorrelated models and developing methods capable of treating noise as it manifests in actual devices---correlated both in space and in time.

\begin{figure*}[t]
\captionsetup[subfloat]{labelformat=empty,captionskip=-24pt}
    \centering
    \subfloat[]{\includegraphics[width=0.85\linewidth]{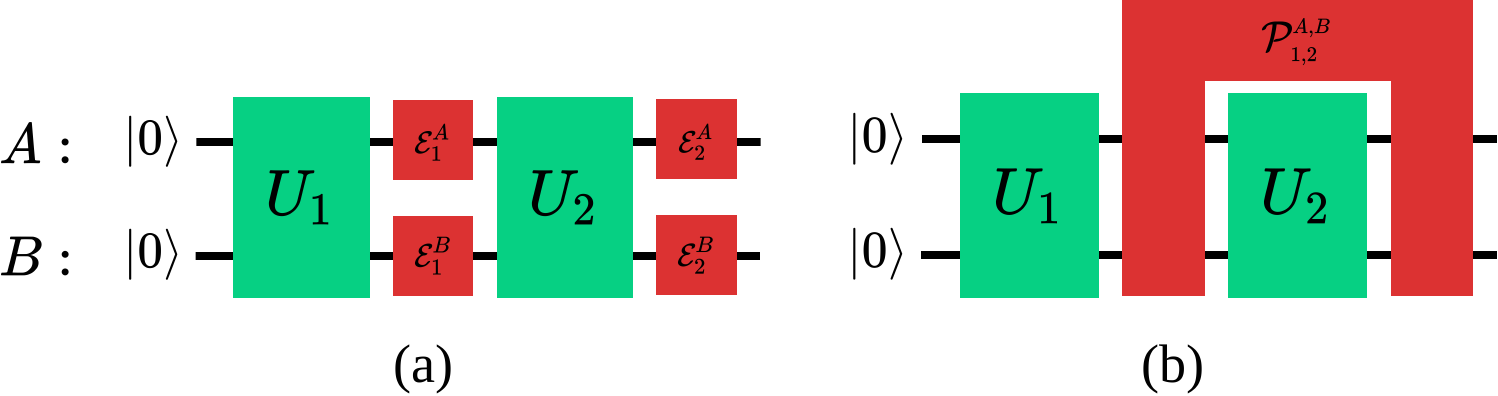}\label{fig:uncorr}}\subfloat[]{\label{fig:corr}}
    \caption{Two noise models for a simple circuit. In both cases, the two qubits are initially prepared in the $\ket{0}$ state and are subsequently acted upon by a first layer of unitaries represented by $U_1$. Following this initial operation, and prior to the subsequent layer $U_2$, the qubits are subject to noise arising from unwanted interactions with their surroundings. The noise occurs again after $U_2$, such that the final state of the qubits is determined by the joint action of the unitaries and the noise processes. (a) The standard approach to this scenario is to assume the noise is uncorrelated both in space and time, being described by local channels. (b) In general, we might have spatiotemporally correlated noise, described by a process tensor.}
    \label{fig:circuits}
\end{figure*}

Although accounting for these correlations complicates noise characterization and suppression, developing mathematical tools that identify, quantify, and constrain them is necessary for making the problem tractable. A natural framework for this purpose is provided by process tensors
\cite{pollock2018non,milz2017introduction,milz2021stochastic,
keeling2026process}. A process tensor describes the evolution of an open quantum system across multiple times by mapping an experimenter's sequence of interventions to the resulting output state. Mathematically, it is a quantum comb
\cite{chiribella2008quantum,chiribella2009theoretical}, whose causal structure retains the information transmitted between different stages of the dynamics through an uncontrolled environment. Markovian dynamics corresponds to a process that factorizes across its
temporal steps \cite{pollock2018operational,milz2019completely}, while a multipartite process tensor can simultaneously retain correlations between different subsystems. Process tensors therefore provide an operational description of general multitime, multisystem noise without requiring the dynamics to be approximated by a sequence of independent channels.

This framework has enabled problems traditionally formulated under Markovian assumptions to be extended to general dynamics with memory. Examples include the characterization and control of noise on quantum
processors
\cite{milz2018reconstructing,white2020demonstration,white2022non,
white2022characterization,white2025what,white2025unifying}, randomized benchmarking in the presence of temporally correlated errors \cite{figueroa2021randomized,figueroa2022towards,figueroa2024operational}, the simulation and learning of non-Markovian dynamics
\cite{jorgensen2019exploiting,jorgensen2020discrete,
cygorek2022simulation,cygorek2024sublinear,guo2020tensor,
fux2023tensor,dowling2024capturing,cygorek2025understanding,
zhang2025learning}, and the formulation of multitime quantum thermodynamics \cite{strasberg2019repeated,huang2022fluctuation,huang2023multiple,zambon2025quantum}. More generally, the process-tensor formalism has supported operational
definitions of quantum Markovianity, information flow, memory strength,
Markov order, and classical versus quantum memory
\cite{milz2019completely,taranto2019quantum,taranto2019structure,
figueroa2019almost,milz2020when,taranto2021non,
sakuldee2022connecting,taranto2023hidden,
taranto2024characterising,santos2025quantifying}. These developments demonstrate that explicitly retaining temporal correlations is not merely a formal generalization, but can qualitatively change how experimentally relevant tasks and dynamical resources are understood.

Recent work has made spatiotemporally correlated noise increasingly tractable from several directions. White \textit{et al.} developed a tensor-network process-tensor framework for classifying non-Markovian effects, reconstructing space-time correlations in quantum processors, and exploiting the resulting models for noise-aware control \cite{white2025unifying}. Oda \textit{et al.} introduced a sparse, experimentally validated model of transmon-based multiqubit operations that captures non-Markovian effects arising from spatiotemporally correlated noise while retaining predictive tractability \cite{oda2026sparse}. Kam \textit{et al.} formulated spatiotemporal Pauli processes, mapping general process-tensor noise under multitime Pauli twirling to correlated Pauli-fault trajectories suitable for scalable simulation and decoding \cite{kam2026spatiotemporal}. Other approaches have focused on the physical nature of temporal memory. In particular, Bäcker \textit{et al.} developed criteria for witnessing quantum memory in spin-boson dynamics \cite{backer2026verifying}, complementing broader investigations of classical and quantum memory in multitime processes \cite{taranto2024characterising,santos2025quantifying}. Together, these works provide tools for reconstructing, modeling, simulating, and witnessing memory effects in correlated quantum dynamics.

The quantification of these correlations, however, remains fragmented. Temporal correlations have commonly been measured by comparing the Choi state of a process with Choi states associated with Markovian processes, whereas spatial correlations have been quantified through distances or information measures relative to product channels or product states. Although these constructions provide useful diagnostics, fixed Choi-state divergences do not, in general, furnish an operational distinguishability measure for quantum processes and need not satisfy data processing under physical transformations of those processes. Here, we address these limitations by introducing a unified, resource-theoretically motivated framework for quantifying temporal, spatial, and total spatiotemporal correlations in general multisystem, multitime processes. Each quantity is defined through optimal distinguishability from the corresponding free set of Markovian, spatially local, or fully uncorrelated processes. Optimization over admissible probing combs ensures a data-processing inequality under every superprocess that preserves the relevant free set, thereby providing an operational quantification layer complementary to existing reconstruction and noise-modeling methods.

We further establish hierarchy and interpolation relations among the three correlation measures and show that fixed Choi-state functionals and other restricted probing strategies provide computable lower bounds. We derive dimension-dependent upper bounds on temporal correlations transmitted through finite classical and quantum memories, together with universal ceilings imposed by the system dimension. Inverting these bounds allows one to certify a minimum memory dimension, exclude bounded-dimensional classical or quantum memory models, or identify an inconsistency in the assumed process description. Finally, we illustrate this certification procedure using effective superconducting-qubit models with \(ZZ\) and \(XY\) interactions, identifying regimes in which the corresponding Choi-state correlations saturate the classical-memory, quantum-memory, and system-dimension ceilings.

The remainder of this work is organized as follows. Section~\ref{sec:framework} introduces the multisystem process-tensor framework and formalizes the sets of Markovian, spatially local, and fully uncorrelated processes. Section~\ref{sec:quantifiers} constructs operational measures of temporal, spatial, and total spatiotemporal correlations, contrasts them with their Choi-state counterparts, and relates restricted probing strategies to the error incurred by neglecting noise correlations under a given class of controls. Section~\ref{sec:bounds} establishes hierarchy and interpolation relations among the correlation measures, derives bounds based on finite-dimensional quantum and classical memory as well as universal system-dimensional ceilings, and formulates their certification consequences. Section~\ref{sec:applications} applies these results to two-step superconducting-qubit models with common and independent environmental memories coupled through \(ZZ\) and \(XY\) interactions. Finally, Sec.~\ref{sec:conclusion} summarizes the implications of the framework and discusses directions toward the experimental characterization and treatment of correlated noise. Supporting proofs and model calculations are provided in the appendices. We begin by formalizing the two noise descriptions contrasted in Fig.~\ref{fig:circuits}.

\section{Framework}\label{sec:framework}

In the uncorrelated model of Fig.~\ref{fig:circuits}, the noise factorizes into local channels $\Ecal_j^X$ across both subsystems and circuit layers. General noise need not admit this factorization and is instead described by a multisystem process tensor $\Pcal_{1,2}^{A,B}$, which we take as the fundamental object of our theory~\cite{pollock2018non,pollock2018operational}.

Consider a scenario where we prepare an initial state $\rho$ at time $t_0$ and then perform a sequence of control operations $\{\Acal_1,\ldots,\Acal_{n-1}\}$ on our system of interest at times $\{t_1,\ldots,t_{n-1}\}$, where each $\Acal_j$ is a completely positive and trace-preserving (CPTP) map, namely, a quantum channel. This goes beyond the typically considered case of unitaries, also encompassing control operations that involve interactions between the system and an ancilla which is later discarded. To determine the final state of the system, we must account for its interaction with an uncontrolled environment between control operations. If the dynamics is Markovian, the evolution between consecutive controls is described by independent quantum channels $\Ecal_j$, which may represent depolarizing, dephasing, or more general types of quantum evolution.

We denote by $\rmi_j$ and $\rmo_j$ the system immediately before and after the $j$th noise block, respectively. Here, $\Dcal(\Hcal_k)$ denotes the set of density operators on the Hilbert space $\Hcal_k$, which we take to be finite-dimensional, and our convention is
\begin{equation}
    \Ecal_j:\Dcal(\Hcal_{\rmi_j})\longrightarrow\Dcal(\Hcal_{\rmo_j}),
    \qquad
    \Acal_j:\Dcal(\Hcal_{\rmo_j})\longrightarrow\Dcal(\Hcal_{\rmi_{j+1}}).
    \label{eq:leg-convention}
\end{equation}
All connected process and control legs are understood to be isomorphic, and step and subsystem labels are suppressed whenever they are clear from context. With this convention, the final state $\rho^{\prime}\in\Dcal(\Hcal_{\rmo_n})$ at time $t_n$ is
\begin{equation}
    \rho^{\prime}=\Ecal_n\circ\Acal_{n-1}\circ\cdots\circ\Acal_1\circ\Ecal_1(\rho),
\end{equation}
as illustrated in Fig.~\subref{fig:mkv}.

Nevertheless, since the system repeatedly interacts with the same environment, information backflows may occur, rendering the resulting dynamics non-Markovian. In this scenario, the evolution cannot be properly described by a set of $n$ independent quantum channels as before. In general, the dynamics is instead characterized by an $n$-step process tensor $\Pcal$, which maps the initial state and the sequence of control operations to the final state of the system,
\begin{equation}
    \rho^{\prime}=\Pcal(\rho,\Acal_1,\ldots,\Acal_{n-1}), 
\end{equation}
as depicted in Fig.~\subref{fig:nmkv}. To accurately characterize the physical properties of the process, $\Pcal$ must be multilinear, completely positive, trace-preserving, and time-ordered. This latter condition reflects the causal structure of the dynamics, where past operations can influence future outputs, but not vice versa. These requirements imply that $\Pcal\in\Pbb_n$, where $\Pbb_n$ denotes the set of $n$-step quantum combs, or simply $n$-combs, following Ref.~\cite{chiribella2009theoretical}. Since the control operations may themselves be correlated through an ancilla, their most general description is provided by a compatible probing comb $\Scal\in\Sbb_n$, where $\Sbb_n$ consists of $n$-step combs without the first input leg, as shown in Fig.~\subref{fig:combs}. We allow $\Scal$ to retain an arbitrary ancilla system $A$, so that the contraction may produce $\Pcal(\Scal)\in\Dcal(\Hcal_{\rmo_n}\otimes\Hcal_A)$. Consequently, the final output is written as
\begin{equation}
    \rho^{\prime}=\Pcal(\Scal),     
\end{equation}
where $\Pcal\in\Pbb_n$ and $\Scal\in\Sbb_n$.

\begin{figure*}[t]
\captionsetup[subfloat]{labelformat=empty,captionskip=-24pt}
    \centering
    \subfloat[]{\includegraphics[width=0.85\linewidth]{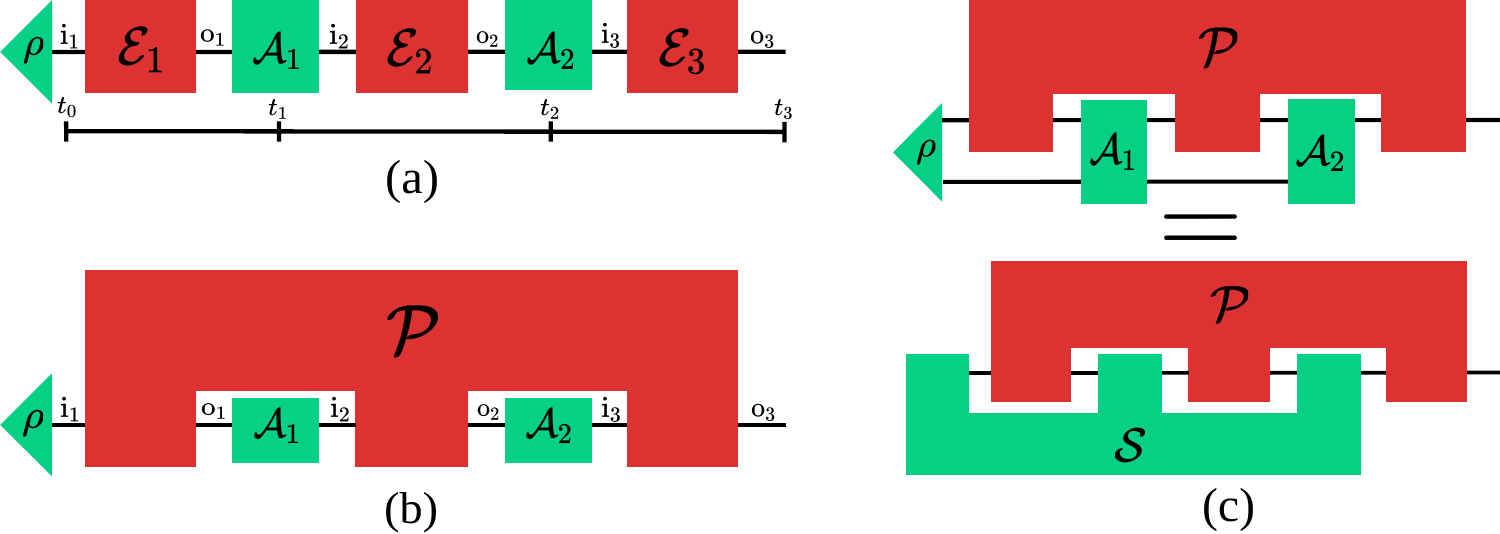}\label{fig:mkv}}
    \subfloat[]{\label{fig:nmkv}}
    \subfloat[]{\label{fig:combs}}
    \caption{(a) Three-step evolution of a quantum system under control operations and Markovian noise. The system is prepared in the state $\rho\in\Dcal(\Hcal_{\rmi_1})$ at time $t_0$ and then is subject to control operations $\Acal_1:\Dcal(\Hcal_{\rmo_1})\to\Dcal(\Hcal_{\rmi_2})$ and $\Acal_2:\Dcal(\Hcal_{\rmo_2})\to\Dcal(\Hcal_{\rmi_3})$ at times $t_1$ and $t_2$, respectively. The interactions with an uncontrolled, yet Markovian, environment between operations are described by the channels $\Ecal_j:\Dcal(\Hcal_{\rmi_j})\to\Dcal(\Hcal_{\rmo_j})$. The final state $\rho^\prime\in\Dcal(\Hcal_{\rmo_3})$ of the system is given by $\rho^{\prime}=\Ecal_3\circ\Acal_{2}\circ\Ecal_2\circ\Acal_1\circ\Ecal_1(\rho)$. (b) Three-step evolution of a quantum system under control operations and non-Markovian noise. Unlike in the Markovian case, the noise process cannot be described by a set of independent channels. In general, it will be given by a process tensor $\Pcal$, mapping the control operations to the final state of the system, that is, $\rho^{\prime}=\Pcal(\rho,\Acal_1,\Acal_{2})$. (c) The physical constraints on multitime dynamics imply any process tensor describing a three-step process must be a three-step quantum comb, i.e., $\Pcal\in\Pbb_3$. Since control operations may also be correlated in time through an ancilla, they are, in general, also described by a comb. In this case, we have $\Scal\in\Sbb_3$. Operationally, $\Scal$ encodes all the control operations and $\Pcal$ encapsulates the effects of the repeated interactions with an uncontrolled environment.}
    \label{fig:pt}
\end{figure*}

Since we are interested in describing dynamics that are also spatially correlated, consider the joint $n$-step evolution of a set of $m$ systems, labeled $\{X_1,\ldots,X_m\}$, whose composite dynamics is described by a comb $\Pcal\in\Pbb_n^m$, where $\Pbb_n^m$ is the set of $n$-step combs acting on $m$ systems. If $\Pcal=\Pcal^{X_1}\otimes\cdots\otimes\Pcal^{X_m}$, where $\Pcal^{X_k}\in\Pbb_n$ is the comb describing the multitime evolution of system $X_k$, then we say the dynamics is spatially uncorrelated (or local), and denote the set of such processes by $\Lbb_{n}^{m}\subset\Pbb_n^m$. Similarly, we denote the set of Markovian $n$-step combs acting on $m$ systems by $\Mbb_n^m\subset\Pbb_n^m$. In the absence of both spatial and temporal correlations, the process $\Pcal$ is described by a sequence of independent quantum channels $\{\Ecal_{1},\ldots,\Ecal_n\}$, where each step factorizes as $\Ecal_{j}=\Ecal_{j}^{X_1}\otimes\cdots\otimes\Ecal_{j}^{X_m}$, that is,
\begin{equation}
    \Pcal
 =\bigotimes_{j=1}^{n}\bigotimes_{k=1}^{m}\Ecal_j^{X_k}.
 \label{eq:free-process}
\end{equation}
Here, the local channel $\Ecal_{j}^{X_k}:\Dcal(\Hcal_{\rmi_j}^k)\to\Dcal(\Hcal_{\rmo_j}^k)$ characterizes the evolution of the $k$-th system between times $t_{j-1}$ and $t_j$. This specific structure represents the idealized case where the dynamics is both Markovian and spatially local, as in Fig.~\subref{fig:uncorr}. We call $\Ubb_n^m\Def(\Mbb_n^m\cap\Lbb_n^m)\subset\Pbb_n^m$ the set of such processes, as they can be used to describe uncorrelated noise. More generally, we could also have dynamics that are spatially local and non-Markovian, spatially correlated and Markovian, or having both spatial and temporal correlations, as in Fig.~\subref{fig:corr}.

Equivalently, an $n$-step process is represented by a positive normalized Choi state
\begin{equation}
    \Upsilon^{\Pcal}\in\Dcal\!\left[\bigotimes_{j=1}^{n}\left(\Hcal_{\rmo_j}\otimes\Hcal_{\rmi_j}\right)\right]
    \label{eq:process-choi-space}
\end{equation}
satisfying the causal normalization conditions of a quantum comb~\cite{chiribella2009theoretical}. For a multipartite system, $\Hcal_{\rmi_j}=\bigotimes_{k=1}^{m}\Hcal_{\rmi_j}^{X_k}$ and similarly for $\Hcal_{\rmo_j}$. Operationally, the Choi state $\Upsilon^{\Pcal}$ is obtained by feeding half of a maximally entangled state into each input leg while retaining the corresponding reference systems and all process outputs, as shown in Fig.~\ref{fig:choi}.

\begin{figure}
    \centering
    \includegraphics[width=0.85\columnwidth]{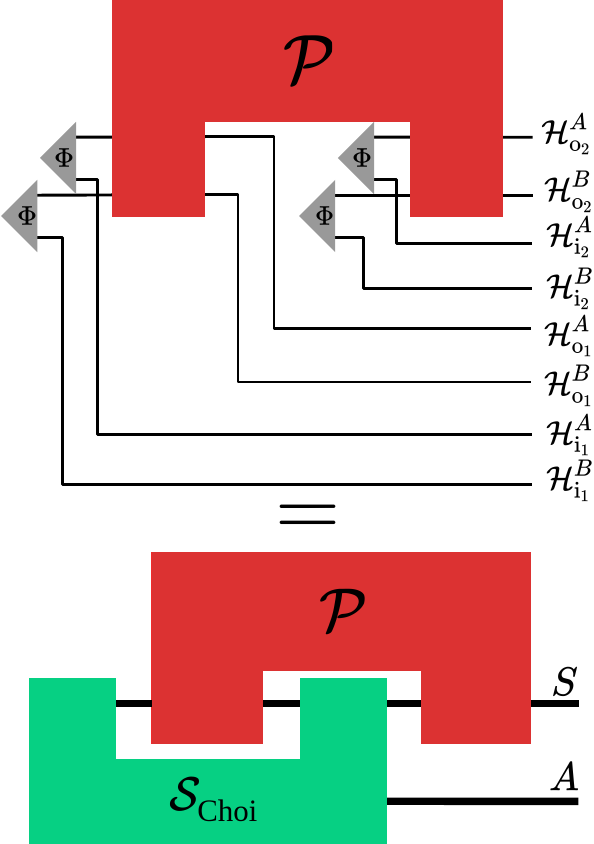}
    \caption{Circuit implementing the Choi state of a process tensor $\Pcal\in\Pbb_2^2$. Each process input is supplied with one half of a maximally entangled state, while the corresponding reference systems and process outputs are retained. The resulting deterministic probing comb $\Scal_{\mathrm{Choi}}\in\Sbb_2$ produces the normalized process Choi state, $\Pcal(\Scal_{\mathrm{Choi}})=\Upsilon^{\Pcal}$. The retained reference systems are included in the output ancilla of the probing comb.}
    \label{fig:choi}
\end{figure}

\section{Correlation quantifiers}\label{sec:quantifiers}

Having established a formal framework to describe general quantum dynamics, we now address the problem of how to meaningfully quantify their spatiotemporal correlations. The approach employed here is to determine how easy it is to distinguish a given process from the sets $\Mbb$, $\Lbb$, or $\Ubb$, according to a chosen measure of state distinguishability. Following Ref.~\cite{zambon2024process}, we formally define this concept below.
\begin{definition}[State divergence]
A state divergence is a mapping $S:\Dcal(\Hcal)\times\Dcal(\Hcal)\to[0,+\infty]$ that is faithful, $S(\rho,\sigma)=0\iff\rho=\sigma$, and obeys data processing under every quantum channel,
\begin{equation}
 S\!\left[\Ecal(\rho),\Ecal(\sigma)\right]\le S(\rho,\sigma).
 \label{eq:state-DPI}
\end{equation}
\end{definition}
This requirement ensures that any measure $S$ quantifies how distinguishable two given states are. Two particular state divergences that we explore extensively throughout this work are the relative entropy, $D\qty(\rho,\sigma)=\tr[\rho(\log \rho - \log\sigma)]$, and the trace distance, $T(\rho,\sigma)=\tr[|\rho-\sigma|]/2$. Throughout, logarithms are taken in base two, and $H(\rho)\Def-\tr(\rho\log\rho)$ denotes the von Neumann entropy. We now investigate two different ways of using state divergences for distinguishing quantum processes.


\subsection{Choi state approach}

A widely used approach for distinguishing quantum processes consists of using some state divergence to distinguish their Choi states, which we formally define below, again following Ref.~\cite{zambon2024process}.
\begin{definition}[Choi divergence]
A Choi divergence $C_{S}(\Pcal,\Qcal)$ of a pair of processes $\Pcal,\Qcal\in\Pbb$ is given by a state divergence $S$ between their Choi states,
    \begin{equation}
        C_{S}(\Pcal,\Qcal)=S(\Upsilon^{\Pcal},\Upsilon^{\Qcal}).
    \end{equation}
\end{definition}
The task of quantifying correlations of some given process $\Pcal$ then reduces to distinguishing its Choi state from those of processes in some set $\Xbb$.
\begin{definition}[Choi correlation functionals]
The Choi temporal correlations $\NS^{\mathrm{Choi}}$ (or Choi non-Markovianity), Choi spatial correlations $\RS^{\mathrm{Choi}}$, and Choi total correlations $\TS^{\mathrm{Choi}}$ are defined as
    \begin{align}
        &\NS^{\mathrm{Choi}}(\Pcal)\Def\CS^{\Mbb}(\Pcal),\\
        &\RS^{\mathrm{Choi}}(\Pcal)\Def\CS^{\Lbb}(\Pcal),\\
        &\TS^{\mathrm{Choi}}(\Pcal)\Def\CS^{\Ubb}(\Pcal),
    \end{align}
where
    \begin{equation}
    \CS^{\Xbb}(\Pcal)= \inf_{\Qcal\in\Xbb}C_{S}(\Pcal,\Qcal)
    \end{equation}
is the Choi distinguishability of process $\Pcal$ with respect to the set $\Xbb$.
\end{definition}
The correlation functionals introduced above unify previous frameworks that utilized Choi-state correlations to measure either exclusively temporal correlations (i.e., non-Markovianity)~\cite{pollock2018operational,figueroa2019almost,berk2021resource,figueroa2021randomized,berk2023extracting,taranto2024characterising,zambon2024relations,zhang2025learning} or exclusively spatial correlations~\cite{rivas2015quantifying,postler2018experimental}. Choi-state measures are often significantly easier to calculate because they fix the probing strategy and may avoid an optimization over general probing combs. For instance, when $S$ is the relative entropy, $\bm{\rmN}_{D}^{\mathrm{Choi}}(\Pcal)$ reduces to the multipartite mutual information among the temporal blocks of $\Upsilon^{\Pcal}$~\cite{berk2021resource}. Their use as resource quantifiers was additionally motivated by the assumption that they are monotonic under operations that do not generate temporal correlations~\cite{berk2021resource}. However, Ref.~\cite{zambon2024process} demonstrated that this monotonicity fails in general for $\bm{\rmN}_{S}^{\mathrm{Choi}}$ and proposed instead the distinguishability-based approach discussed next.

Similarly, the spatial correlation measure introduced in Ref.~\cite{rivas2015quantifying}, which up to normalization is equivalent to $\bm{\rmR}_{D}^{\mathrm{Choi}}$ for our definitions, relies on monotonicity under superchannels that do not generate spatial correlations. Nevertheless, in Appendix~\ref{app:counterexample}, we adapt the counterexample from Ref.~\cite{zambon2024process} to show that this spatial measure suffers from the exact same pathology; contrary to the claims in Ref.~\cite{rivas2015quantifying}, the measure $\bm{\rmR}_{D}^{\mathrm{Choi}}$ does not fulfill the monotonicity property. Consequently, the total Choi correlation functional $\TS^{\mathrm{Choi}}$ is likewise nonmonotonic under superprocesses that preserve the set of uncorrelated processes. That is, there exist physically admissible maps $\bm{Z}$ on combs satisfying $\bm{Z}(\Ubb_n^m)\subseteq\Ubb_{n'}^{m'}$ such that $\TS^{\mathrm{Choi}}\!\left[\bm{Z}(\Pcal)\right] >\TS^{\mathrm{Choi}}(\Pcal)$, an undesirable feature for any correlation quantifier motivated by a resource-theoretic framework. Yet, despite lacking this important property of monotonicity, in Sec.~\ref{sec:bounds} we show next that these correlations in the Choi state remain highly relevant, as they provide useful bounds on the valid quantifiers that do fulfill this condition.

\subsection{Distinguishability approach}

We now adopt an alternative approach to quantify correlations of general quantum processes. First, we note that the fundamental limitation of quantifiers based on Choi divergences stems from their lack of monotonicity under the action of superprocesses. In Ref.~\cite{zambon2024process}, this drawback is circumvented by employing generalized divergences instead, which we formally define below.
\begin{definition}[Generalized divergence]
A generalized divergence $G_{S}(\Pcal,\Qcal)$ of a pair of processes $\Pcal,\Qcal\in\Pbb$ is given by a state divergence $S$ between their outputs when acting on the same comb $\Scal\in\Sbb$, maximized over all compatible combs,
    \begin{equation}
        G_{S}(\Pcal,\Qcal)=\sup_{\Scal\in\Sbb}S\!\left[\Pcal(\Scal),\Qcal(\Scal)\right].
    \end{equation}
\end{definition}
Operationally, $G_{S}$ measures how distinguishable the outputs of $\Pcal$ and $\Qcal$ can be made given an optimal control strategy $\Scal$ designed to differentiate them. Alternatively, if one possesses the process $\Pcal$ but models it as if it were $\Qcal$, $G_{S}$ quantifies how far the actual output deviates from the expected one in the worst-case scenario. For example, if we consider a pair $(\Ecal,\Fcal)$ of one-step quantum combs, i.e. quantum channels, then $G_T(\Ecal,\Fcal)$ is the diamond distance between them~\cite{watrous2018theory}. We now use the generalized divergence to define the correlations of a process.
\begin{definition}[Correlations of a process]
The temporal correlations $\NS$ (or non-Markovianity), spatial correlations $\RS$, and total correlations $\TS$ are defined as
    \begin{align}
        &\NS(\Pcal)\Def\GS^{\Mbb}(\Pcal),\\
        &\RS(\Pcal)\Def\GS^{\Lbb}(\Pcal),\\
        &\TS(\Pcal)\Def\GS^{\Ubb}(\Pcal),
    \end{align}
where
    \begin{equation}
    \GS^{\Xbb}(\Pcal)= \inf_{\Qcal\in\Xbb}G_{S}(\Pcal,\Qcal)
    \end{equation}
is the generalized distinguishability $\GS^{\Xbb}(\Pcal)$ of process $\Pcal$ with respect to the set $\Xbb$.
\end{definition}
The monotonicity of generalized divergences under the action of superprocesses~\cite{zambon2024process} directly implies the monotonicity of $\bm{\rmG}_{S}^{\Xbb}$ under superprocesses that do not generate correlations that are not present in $\Xbb$. We formally state this property below as a theorem, the proof of which is provided in Appendix~\ref{app:monotonicity}.
\begin{theorem}[Monotonicity of $\bm{\rmG}_{S}^{\Xbb}$]\label{thm:monotonicity}
Let $\bm{Z}$ be a superprocess satisfying $\bm{Z}(\Xbb)\subseteq\Xbb$. Then, for any $\Pcal\in\Pbb$ it holds
    \begin{equation}
        \bm{\rmG}_{S}^{\Xbb}\!\left[\bm{Z}(\Pcal)\right]\le \bm{\rmG}_{S}^{\Xbb}(\Pcal).
    \end{equation}
\end{theorem}

\subsection{Correlation error under restricted control}

We now go back to the original problem of quantifying the impact of spatiotemporally correlated noise in quantum circuits, which are typically modeled as a sequence of independent layers of unitary gates. Let $\Cbb_n^m \subset \Sbb_n^m$ denote the subset of combs in $\Sbb_n^m$ where the first leg initializes a pure state $\ket{\psi_0} \in \Hcal_{\text{i}_1}^{X_1} \otimes \cdots \otimes \Hcal_{\text{i}_1}^{X_m}$, and each of the subsequent $n-1$ steps consists of a joint unitary operation over all $m$ qubits, with no auxiliary systems involved. In this manner, the elements of $\Cbb_n^m$ characterize all possible circuits composed of $n-1$ layers acting on $m$ qubits, along with all possible pure initial states. Analogously, the set of all potential noise processes affecting such circuits is described by the combs in $\Pbb_n^m$.

Now, suppose the actual noise in a given platform is described by a process tensor $\Pcal \in \Pbb_n^m$, but one instead assumes it to be uncorrelated, modeling it via some $\Qcal \in \Ubb_n^m$. We can then address the question of how much the actual result of the computation can deviate, according to a chosen state divergence $S$, from the expected outcome in the worst-case scenario over all possible circuits $\Ccal\in \Cbb_n^m$, assuming that the best possible uncorrelated noise approximation was chosen. This deviation is captured by the correlation error for ideal (unitary) control $\bm{\rmE}_{S}^{\Cbb}$, which we formally define below.
\begin{definition}[Correlation error for ideal control]
The correlation error for ideal control $\bm{\rmE}_{S}^{\Cbb}(\Pcal)$ of a process $\Pcal\in\Pbb$, for a given state divergence $S$, is given by
    \begin{equation}
        \bm{\rmE}_{S}^{\Cbb}(\Pcal)= \inf_{\Qcal\in\Ubb}\sup_{\Ccal\in\Cbb}S\!\left[\Pcal(\Ccal),\Qcal(\Ccal)\right].
    \end{equation}
\end{definition}
\noindent Operationally, this quantity gives the maximum error one can incur in a computation by ignoring the spatiotemporal correlations of the noise. This is achieved by assessing how well $\Pcal$ can be distinguished from its optimal uncorrelated approximation when restricted exclusively to unitary circuits.

Nevertheless, due to intrinsic hardware limitations, experimental quantum circuits are never perfectly described by a sequence of ideal, independent unitaries. This implies that the actual control comb $\Scal$ may lie outside the set $\Cbb$, potentially leading to computational errors that exceed $\bm{\rmE}_{S}^{\Cbb}$. To account for the error arising from noise correlations when the control operations themselves are noisy, we may change the maximization in the definition of $\bm{\rmE}_S^{\Cbb}$ to some set $\Ybb\subseteq\Sbb$ describing the actual available control. This yields the correlation error for $\Ybb$ control $\bm{\rmE}_{S}^{\Ybb}$, defined below.
\begin{definition}[Correlation error for $\Ybb$ control]
The correlation error for $\Ybb$ control $\bm{\rmE}_{S}^{\Ybb}(\Pcal)$ of a process $\Pcal\in\Pbb$, for a given state divergence $S$, is given by
    \begin{equation}
        \bm{\rmE}_{S}^{\Ybb}(\Pcal)= \inf_{\Qcal\in\Ubb}\sup_{\Ycal\in\Ybb}S\!\left[\Pcal(\Ycal),\Qcal(\Ycal)\right].
    \end{equation}
\end{definition}
\noindent Since both $\Ybb$ and $\Cbb$ are subsets of $\Sbb$, we have $\ES^{\Ybb}(\Pcal)\le\TS(\Pcal)$ and $\ES^{\Cbb}(\Pcal)\le\TS(\Pcal)$, that is, the errors under restricted control are upper bounded by the total correlations of the process. Moreover, let $\Ybb=\{\Scal_{\mathrm{Choi}}\}$, since the Choi probe is one particular choice we have
\begin{equation}
 \TS^{\mathrm{Choi}}(\Pcal)=\ES^{\{\Scal_{\mathrm{Choi}}\}}(\Pcal)\le\TS(\Pcal).
 \label{eq:choi-lower-bound}
\end{equation}
These lower bounds to $\TS$ are useful both for being easier to calculate and for carrying important practical significance. Furthermore, if any of them reaches a valid upper bound on $\TS$, it certifies the exact value of the correlation, as we show in the example of Sec.~\ref{sec:applications}. Having defined these useful quantifiers, we now turn to the derivation of useful bounds regarding them.

\section{Bounds on spatiotemporal correlations}\label{sec:bounds}

While the quantifiers we defined in Sec.~\ref{sec:quantifiers} satisfy important properties and carry relevant physical significance, it may be challenging to calculate them in practical scenarios, both operationally and numerically. Thus, we now derive a series of useful analytical bounds that may help circumvent this caveat in practical applications.

\subsection{Immediate lower bounds}

We begin by stating a set of lower bounds to the correlation quantifiers. Despite being immediate from their definitions, these bounds are especially informative when they saturate one of the upper bounds we derive later, helping certify that the quantity they are bounding from below is actually saturating its related upper bound, while being easier to calculate than the actual quantity.

\begin{proposition}[Correlation hierarchy]\label{prop:hierarchy}
For every process $\Pcal$ and every divergence used in the definitions,
\begin{equation}
 \max\{\NS(\Pcal),\RS(\Pcal)\}\le\TS(\Pcal).
 \label{eq:operational-lower-hierarchy}
\end{equation}
The same ordering holds for the Choi quantifiers,
\begin{equation}
 \max\{\NS^{\mathrm{Choi}}(\Pcal),\RS^{\mathrm{Choi}}(\Pcal)\}
 \le \TS^{\mathrm{Choi}}(\Pcal).
 \label{eq:choi-lower-hierarchy}
\end{equation}
Moreover, restriction to any fixed class of probing combs gives a lower bound on the corresponding operational quantifier.  In particular,
\begin{equation}
 \NS^{\mathrm{Choi}}\le\NS,\qquad
 \RS^{\mathrm{Choi}}\le\RS,\qquad
 \TS^{\mathrm{Choi}}\le\TS.
 \label{eq:choi-operational-order}
\end{equation}
\end{proposition}
The proof is given in App.~\ref{app:quantifier-relations}.

\subsection{Interpolation bounds}

To upper-bound total correlations in terms of their temporal and spatial parts, data processing alone is insufficient. The following result uses the triangle inequality and tensor subadditivity of trace distance.

\begin{proposition}[Trace-distance interpolation bound]\label{prop:trace-interpolation}
For any $\Pcal\in\Pbb_n^m$ it holds,
\begin{equation}
 \bm{\rmT}_{T}(\Pcal)
 \le \bm{\rmN}_{T}(\Pcal)+n\qty[\bm{\rmN}_{T}(\Pcal)+\bm{\rmR}_{T}(\Pcal)].
 \label{eq:trace-interpolation-n}
\end{equation}
Interchanging the temporal and spatial factorizations also gives
\begin{equation}
 \bm{\rmT}_{T}(\Pcal)
 \le m\qty[\bm{\rmN}_{T}(\Pcal)+\bm{\rmR}_{T}(\Pcal)]+\bm{\rmR}_{T}(\Pcal).
 \label{eq:trace-interpolation-m}
\end{equation}
\end{proposition}
The proof is given in App.~\ref{app:trace-interpolation}. Together with Proposition~\ref{prop:hierarchy}, this shows that simultaneous small temporal and spatial correlations force the total trace-distance correlation to be small.

Relative entropy has no triangle inequality, so the preceding operational proof does not apply. For normalized Choi states, however, the relative-entropy projection onto each product set is the product of the corresponding marginals, which yields the following interpolation bounds.

\begin{proposition}[Relative-entropy Choi interpolation bound]\label{prop:choi-interpolation}
For any $\Pcal\in\Pbb_n^m$ it holds,
\begin{equation}
 \bm{\rmT}_{D}^{\mathrm{Choi}}(\Pcal)
 \le \bm{\rmN}_{D}^{\mathrm{Choi}}(\Pcal)+n\bm{\rmR}_{D}^{\mathrm{Choi}}(\Pcal).
 \label{eq:choi-interpolation-n}
\end{equation}
Equivalently, exchanging the two partitions,
\begin{equation}
 \bm{\rmT}_{D}^{\mathrm{Choi}}(\Pcal)
 \le \bm{\rmR}_{D}^{\mathrm{Choi}}(\Pcal)+m\bm{\rmN}_{D}^{\mathrm{Choi}}(\Pcal).
 \label{eq:choi-interpolation-m}
\end{equation}
\end{proposition}
The proof is given in App.~\ref{app:choi-interpolation}.

\subsection{Finite-dimensional memory bounds}

Next, we use the fact that the environment dimension is finite to bound the temporal correlations (non-Markovianity). Any process comb admits a sequential realization~\cite{chiribella2009theoretical} with memory systems $E_0,\ldots,E_n$ and CPTP maps
\begin{equation}
    \Ecal_j:\Dcal(\Hcal_{\rmi_j}\otimes\Hcal_{E_{j-1}})\longrightarrow\Dcal(\Hcal_{\rmo_j}\otimes\Hcal_{E_j}),
    \qquad j=1,\ldots,n.
    \label{eq:memory-channel-realization}
\end{equation}
The initial memory is incorporated into $\Ecal_1$, the final memory is discarded, and $E_1,\ldots,E_{n-1}$ form the $n-1$ internal temporal memory links. Accordingly, the results below apply whenever a comb admits a realization of this form with bounded inter-step memory, irrespective of whether that memory is identified with a microscopic environment.

The following result follows from a binary preparation, or reverse-test, construction. Closely related constructions underlie the exact distinguishability-dilution characterization of the max-relative entropy \cite{wang2019resource} and, more generally, the maximal extension of classical divergences to quantum states \cite{gour2020optimal}. We state the result in a form applicable to any state divergence satisfying data processing.
\begin{lemma}[State domination]\label{lem:state-domination}
If $\lambda\rho\le\sigma$ for normalized states and $\lambda\in[0,1]$, then
\begin{equation}
 S(\rho,\sigma)\le b_S(\lambda),
 \label{eq:state-domination-app}
\end{equation}
with the nonincreasing calibration function defined as
\begin{equation}
 b_S(\lambda)\Def S\!\left[\ketbra{0}{0},\lambda\ketbra{0}{0}+(1-\lambda)\ketbra{1}{1} \right].
 \label{eq:binary-calibration}
\end{equation}
In particular, we have $b_D(\lambda)=-\log\lambda$ and $b_T(\lambda)=1-\lambda$.
\end{lemma}
Equivalently, \(\sigma=\lambda\rho+(1-\lambda)\tau\) for some state
\(\tau\), so \(\lambda\) is the probability with which \(\rho\) occurs
in the binary preparation of \(\sigma\). This lemma is proved in
App.~\ref{app:domination}. In the proofs of the dimension bounds below,
we apply it after constructing a free process \(\Qcal\) satisfying
\(\lambda\Upsilon^{\Pcal}\le\Upsilon^{\Qcal}\); the corresponding
process-level domination result is stated and proved in
App.~\ref{app:free-candidate-domination}.

\subsubsection{Quantum memory}

We first consider the case where the memory may be fully quantum.

\begin{theorem}[Finite-memory bound on temporal correlations]
\label{thm:memory-bound}
Let a process $\Pcal\in\Pbb_n$ have a dilation using $n-1$ environmental memory systems, as described in Eq.~\eqref{eq:memory-channel-realization}, with dimensions $d_j=\dim{\Hcal_{E_j}}$. Then, it holds
\begin{equation}
 \NS(\Pcal)
 \le b_S\!\left(\prod_{j=1}^{n-1}d_j^{-2}\right).
 \label{eq:memory-profile-bound}
\end{equation}
In particular, if $d_j\le d_E$ at every link,
\begin{equation}
 \NS(\Pcal)\le b_S\!\left[d_E^{-2(n-1)}\right].
 \label{eq:environment-bound}
\end{equation}
The same bound holds for $\NS^{\mathrm{Choi}}$ by Eq.~\eqref{eq:choi-operational-order}.
\end{theorem}
This is proved together with the classical refinement in App.~\ref{app:memory-bounds}. For the two principal divergences,
\begin{align}
 \bm{\rmN}_{D}(\Pcal)&\le2(n-1)\log d_E,
 \label{eq:D-quantum-memory-bound}\\
 \bm{\rmN}_{T}(\Pcal)&\le1-d_E^{-2(n-1)}.
 \label{eq:T-quantum-memory-bound}
\end{align}
This bound may be interpreted as the dimension of the memory limiting how much the correlations between different time steps may grow.

\subsubsection{Classical memory}

Sometimes, the process may be dilated with a memory that is classical, as we define below.
\begin{definition}[Classical memory]\label{def:classical-memory}
A realization of Eq.~\eqref{eq:memory-channel-realization} has classical memory if, for every inter-step memory $E_j$, there is a fixed basis $\{\ket a\}_{a=1}^{d_j}$ such that inserting the complete dephasing channel
\begin{equation}
 \Zcal_{E_j}(X)
 =\sum_{a=1}^{d_j}\ketbra{a}{a}X\ketbra{a}{a}
 \label{eq:dephasing-channel}
\end{equation}
simultaneously on all inter-step memory links leaves the process unchanged for every compatible probing comb.
\end{definition}
When the memory is classical, the finite-dimension bounds may be refined.
\begin{theorem}[Classical-memory refinement]
\label{thm:classical-memory-bound}
If $\Pcal\in\Pbb_n$ admits the realization of Definition~\ref{def:classical-memory}, then
\begin{equation}
 \NS(\Pcal)
 \le b_S\!\left(\prod_{j=1}^{n-1}d_j^{-1}\right).
 \label{eq:classical-profile-bound}
\end{equation}
If $d_j\le d_E$ for every cut,
\begin{equation}
 \NS(\Pcal)\le b_S\!\left[d_E^{-(n-1)}\right].
 \label{eq:classical-environment-bound}
\end{equation}
Again, $\NS^{\mathrm{Choi}}$ obeys the same ceiling.
\end{theorem}
This is proved together with the quantum-memory bound in App.~\ref{app:memory-bounds}. For the two principal divergences,
\begin{align}
 \bm{\rmN}_{D}(\Pcal)&\le(n-1)\log d_E,
 \label{eq:D-classical-memory-bound}\\
 \bm{\rmN}_{T}(\Pcal)&\le1-d_E^{-(n-1)}.
 \label{eq:T-classical-memory-bound}
\end{align}
Again, the memory dimension limits the scaling, while a classical link can transmit less correlation than a quantum link of the same dimension.

\subsection{Universal system ceilings}
\label{subsec:system-ceilings}

The memory bounds derived above constrain temporal correlations whenever a finite-dimensional memory realization is available.  We now derive independent ceilings that hold for every valid process tensor and require no assumption about its environmental realization.  Let \(d_S\) be the dimension of the complete system carrier at each step, including all of its spatial subsystems, and suppose for simplicity that all input and output carriers have this same dimension.

\begin{theorem}[Universal system ceilings]
\label{thm:system-ceilings}
For any $\Pcal\in\Pbb_n$ with system dimension $d_S$ it holds,
\begin{align}
 \NS(\Pcal)
 &\le b_S\!\left[d_S^{-2(n-1)}\right],
 \label{eq:N-system-ceiling}\\
 \RS(\Pcal)
 &\le b_S\!\left(d_S^{-2n}\right),
 \label{eq:R-system-ceiling}\\
 \TS(\Pcal)
 &\le b_S\!\left(d_S^{-2n}\right).
 \label{eq:T-system-ceiling}
\end{align}
The same bounds hold for the corresponding Choi quantifiers.  Moreover, for every restricted class of probing combs \(\Ybb\subseteq\Sbb\),
\begin{equation}
 \ES^{\Ybb}(\Pcal)
 \le\TS(\Pcal)
 \le b_S\!\left(d_S^{-2n}\right).
 \label{eq:E-system-ceiling}
\end{equation}
\end{theorem}
This is proved in App.~\ref{app:system-ceilings}. The different exponents in Theorem~\ref{thm:system-ceilings} have a direct structural origin.  To construct a Markovian approximation, the first one-step block may remain arbitrary, and only the remaining \(n-1\) temporal blocks need to be replaced.  That first block may, however, still contain spatial correlations.  Obtaining a spatially local or fully uncorrelated candidate therefore requires replacing all \(n\) blocks, producing one additional factor \(d_S^2\).

For relative entropy, Theorem~\ref{thm:system-ceilings} becomes
\begin{align}
 \bm{\mathrm N}_{D}(\Pcal)&\le2(n-1)\log d_S,
 \label{eq:ND-system-ceiling}\\
 \bm{\mathrm R}_{D}(\Pcal)&\le2n\log d_S,
 \label{eq:RD-system-ceiling}\\
 \bm{\mathrm T}_{D}(\Pcal)&\le2n\log d_S,
 \label{eq:TD-system-ceiling}
\end{align}
where Eq.~\eqref{eq:ND-system-ceiling} generalizes Thm.~1 of Ref.~\cite{zambon2024relations} to the optimized operational quantifier, whereas trace distance gives
\begin{align}
 \bm{\mathrm N}_{T}(\Pcal)&\le1-d_S^{-2(n-1)},
 \label{eq:NT-system-ceiling}\\
 \bm{\mathrm R}_{T}(\Pcal)&\le1-d_S^{-2n},
 \label{eq:RT-system-ceiling}\\
 \bm{\mathrm T}_{T}(\Pcal)&\le1-d_S^{-2n}.
 \label{eq:TT-system-ceiling}
\end{align}
Combining Eq.~\eqref{eq:N-system-ceiling} with the quantum-memory and classical-memory bounds derived above yields, respectively,
\begin{align}
 \NS(\Pcal)
 &\le b_S\!\left[\min\{d_E,d_S\}^{-2(n-1)}\right],
 \label{eq:N-combined-quantum}\\
 \NS(\Pcal)
 &\le b_S\!\left[\min\{d_E,d_S^2\}^{-(n-1)}\right].
 \label{eq:N-combined-classical}
\end{align}
Consequently, for relative entropy and trace distance an arbitrary quantum-memory bound is stronger than the universal system ceiling precisely in the regime \(d_E<d_S\), while a classical-memory bound remains informative throughout the larger regime \(d_E<d_S^2\).  In contrast, the dimension of the temporal memory alone does not generally constrain \(\RS\) or \(\TS\), since each one-step block may contain arbitrary spatial correlations.

\subsection{Certification consequences}
\label{subsec:certification}

The preceding inequalities may be inverted into certification statements. Let \(W_S^{\rm N}\), \(W_S^{\rm R}\), and \(W_S^{\rm T}\) be experimentally or numerically certified lower bounds satisfying
\begin{equation}
 W_S^{\rm N}\le\NS(\Pcal),\qquad 
 W_S^{\rm R}\le\RS(\Pcal),\qquad
 W_S^{\rm T}\le\TS(\Pcal).
 \label{eq:certified-lower-bounds}
\end{equation}
They may be obtained from any single probing comb or any restricted probing class; the Choi probe is only one particular choice.  Define the generalized inverse of the nonincreasing calibration function by
\begin{equation}
 b_S^{\leftarrow}(w)
 \Def\sup\{\lambda\in[0,1]:b_S(\lambda)\ge w\}.
 \label{eq:b-generalized-inverse}
\end{equation}

\begin{corollary}[Memory and system-dimension certification]
\label{cor:certification}
For \(n\ge2\), a certified temporal value \(W_S^{\rm N}\) implies, for arbitrary quantum and classical memory respectively,
\begin{align}
 d_E&\ge
 \left\lceil
 \left[b_S^{\leftarrow}(W_S^{\rm N})\right]^{-1/[2(n-1)]}
 \right\rceil,
 \label{eq:quantum-dimension-witness}\\
 d_E&\ge
 \left\lceil
 \left[b_S^{\leftarrow}(W_S^{\rm N})\right]^{-1/(n-1)}
 \right\rceil.
 \label{eq:classical-dimension-witness}
\end{align}
The system ceilings similarly imply
\begin{align}
 d_S&\ge
 \left\lceil
 \left[b_S^{\leftarrow}(W_S^{\rm N})\right]^{-1/[2(n-1)]}
 \right\rceil,
 \label{eq:system-N-witness}\\
 d_S&\ge
 \left\lceil
 \left[b_S^{\leftarrow}(W_S^{\rm C})\right]^{-1/(2n)}
 \right\rceil,
 \qquad \mathrm C\in\{\mathrm R,\mathrm T\}.
 \label{eq:system-RT-witness}
\end{align}
\end{corollary}

This is proved in App.~\ref{app:certification}. For relative entropy, Eqs.~\eqref{eq:quantum-dimension-witness} and \eqref{eq:classical-dimension-witness} read
\begin{align}
 d_E&\ge
 \left\lceil2^{W_D^{\rm N}/[2(n-1)]}\right\rceil,
 \label{eq:D-quantum-dimension-witness}\\
 d_E&\ge
 \left\lceil2^{W_D^{\rm N}/(n-1)}\right\rceil,
 \label{eq:D-classical-dimension-witness}
\end{align}
while for trace distance, whenever \(W_T^{\rm N}<1\), they become
\begin{align}
 d_E&\ge
 \left\lceil(1-W_T^{\rm N})^{-1/[2(n-1)]}\right\rceil,
 \label{eq:trace-quantum-dimension-witness}\\
 d_E&\ge
 \left\lceil(1-W_T^{\rm N})^{-1/(n-1)}\right\rceil.
 \label{eq:trace-classical-dimension-witness}
\end{align}

Suppose that independent physical information gives an upper bound \(d_E\le d_{\max}\).  Then
\begin{equation}
 W_S^{\rm N}>b_S\!\left[d_{\max}^{-(n-1)}\right]
 \label{eq:general-nonclassicality-witness}
\end{equation}
rules out every classical-memory realization of dimension at most \(d_{\max}\).  In particular,
\begin{align}
 W_D^{\rm N}&>(n-1)\log d_{\max},
 \label{eq:D-nonclassicality-witness}\\
 W_T^{\rm N}&>1-d_{\max}^{-(n-1)}
 \label{eq:trace-nonclassicality-witness}
\end{align}
are relative-entropy and trace-distance witnesses of nonclassical temporal memory under this dimension assumption.  Likewise,
\begin{equation}
 W_S^{\rm N}>b_S\!\left[d_{\max}^{-2(n-1)}\right]
 \label{eq:general-quantum-memory-exclusion}
\end{equation}
rules out every quantum-memory realization of dimension at most \(d_{\max}\).

The system ceilings have a different interpretation.  If the nominal system dimension is known to be \(d_S\), any violation of Eqs.~\eqref{eq:N-system-ceiling}--\eqref{eq:T-system-ceiling} cannot be explained by increasing the environment dimension.  It instead falsifies at least one assumption entering the process model, such as the chosen system Hilbert space, the number or placement of temporal steps, or the validity of the certified lower bound.  Leakage outside the computational subspace is a particularly relevant physical mechanism.  Finally, neither \(W_S^{\rm R}\) nor \(W_S^{\rm T}\) alone lower bounds \(d_E\) without additional structural assumptions: purely spatial correlations may already be present inside a single temporal block.

\section{Superconducting-qubit memory models}
\label{sec:applications}

We now illustrate the memory bounds in effective models relevant to superconducting devices.  The system consists of two qubits, labeled $A$ and $B$, and the same system--environment interaction is applied in two consecutive temporal blocks, with an arbitrary control operation allowed between them, implying the noise will be described by a comb $\Pcal_{\theta}\in\Pbb_{2}^{2}$, where $\theta$ is a tunable parameter of the dynamics, determined by the time interval between operations.  We evaluate the relative-entropy Choi temporal correlations
\begin{equation}
 \bm{\rmN}_{D}^{\mathrm{Choi}}(\Pcal_\theta)
 =I(T_1:T_2)_{\Upsilon^{\Pcal_\theta}},
 \label{eq:application-choi-N}
\end{equation}
where $T_j$ contains all normalized-Choi legs associated with temporal block $j$.  Since the Choi preparation is an admissible probing comb,
\begin{equation}
 \bm{\rmN}_{D}^{\mathrm{Choi}}(\Pcal_\theta)
 \le \bm{\rmN}_{D}(\Pcal_\theta).
 \label{eq:application-choi-lower}
\end{equation}
It therefore provides an experimentally meaningful lower bound on the operational temporal correlations.

For two system qubits, $d_S=4$, the system ceiling for a two-step process is
\begin{equation}
 \bm{\rmN}_{D}(\Pcal_\theta)\le2\log d_S=4.
 \label{eq:application-system-ceiling}
\end{equation}
Unlike the environment-dependent bounds, Eq.~\eqref{eq:application-system-ceiling} follows solely from comb causality and holds independently of the dimension, structure, and classicality of the environment.  It is therefore the absolute temporal-correlation ceiling compatible with the assumed two-step process on the two-qubit system Hilbert space. The quantum- and classical-memory bounds for the two environment sizes considered below are:
\begin{center}
\begin{tabular}{c@{\qquad}c@{\qquad}c}
\toprule
$d_E$ & quantum memory & classical memory\\
\midrule
$2$ & $2\log d_E=2$ & $\log d_E=1$\\
$4$ & $2\log d_E=4$ & $\log d_E=2$\\
\bottomrule
\end{tabular}
\end{center}
Thus, for a single environmental qubit both bounds are strictly stronger than the system ceiling, whereas a two-qubit quantum environment has enough memory capacity to reach it.

\subsection{A common environmental qubit}
\label{subsec:common-environment}

Consider first a single persistent environmental qubit \(E\) coupled to both \(A\) and \(B\). Effective dispersive cross-Kerr and resonant-exchange reductions of circuit-QED and transmon--TLS Hamiltonians motivate the following models \cite{blais2021circuit,zuk2024robust}. In a rotating frame, after removing local frequency shifts, we write
\begin{align}
 H_{ZZ}^{(1E)}
 &=J_A Z_AZ_E+J_B Z_BZ_E,
 \label{eq:common-ZZ}\\
 H_{XY}^{(1E)}
 &=g_A\left(\sigma_A^+\sigma_E^-+\sigma_A^-\sigma_E^+\right)\nonumber\\
 &\quad+g_B\left(\sigma_B^+\sigma_E^-+\sigma_B^-\sigma_E^+\right).
 \label{eq:common-XY}
\end{align}
The first Hamiltonian follows from an effective dispersive interaction \(\chi_A Z_A n_E+\chi_B Z_B n_E\): restricting \(E\) to its two lowest states, using \(n_E=(\mathbb I_E\mp Z_E)/2\), and absorbing the resulting local \(Z_A\) and \(Z_B\) shifts, signs, and numerical factors into \(J_A\) and \(J_B\) gives Eq.~\eqref{eq:common-ZZ}. The second Hamiltonian is the rotating-wave, two-level reduction of near-resonant transverse excitation exchange. These effective reductions are applicable when \(E\) is a parasitic two-level system or a sufficiently anharmonic ancillary qubit or coupler, as supported experimentally for coherent defects in superconducting qubits \cite{liu2024observation}. A nearly harmonic circuit mode admits the same two-level description only when population outside its two lowest levels remains negligible. Static \(ZZ\) crosstalk and tunable exchange interactions occur in superconducting circuits \cite{mundada2019suppression,zhao2021zz,sung2021iswap,foxen2020fsim}. We take symmetric couplings \(J_A=J_B=J\) and \(g_A=g_B=g\), and write the interaction angle as \(\theta=Jt\) or \(\theta=gt\).

For Eq.~\eqref{eq:common-ZZ}, we initialize the environment in \(\mathbb I_E/2\). This state represents an equal-population classical binary fluctuator, a driven or intentionally randomized ancillary degree of freedom, or a correlation stress test. The interaction commutes with \(Z_E\), so dephasing \(E\) in the computational basis at the temporal cut leaves the process invariant. The transmitted memory is therefore classical and
\begin{equation}
 \bm{\rmN}_{D}^{\mathrm{Choi}}(\Pcal_\theta)
 \le\bm{\rmN}_{D}(\Pcal_\theta)\le1,
 \label{eq:common-ZZ-bound}
\end{equation}
as shown by the solid blue curve in Fig.~\subref{fig:classical-ZZ}. The Choi quantity reaches one bit at $\theta=\pi/4$, proving that the physical common-$ZZ$ model saturates the classical single-qubit-memory bound.

For Eq.~\eqref{eq:common-XY}, we take the initial state $\ketbra{0}{0}_E$.  Excitation exchange coherently transfers quantum information through $E$, and the classical-memory restriction no longer applies.  The resulting Choi temporal correlations cross the classical ceiling of one bit and approach, but do not attain, the quantum ceiling of two bits, which is illustrated by the solid blue curve of Fig.~\subref{fig:quantum-XY}.  Consequently,
\begin{equation}
 \bm{\rmN}_{D}^{\mathrm{Choi}}(\Pcal_\theta)>1
 \label{eq:common-nonclassical-witness}
\end{equation}
certifies that no classical single-qubit memory can reproduce the process. This is a one-sided witness: values below one do not imply that the memory is classical.

\subsection{Increasing the memory dimension}
\label{subsec:two-environments}

To expose the dimension dependence, let $E=E_AE_B$ consist of two persistent environmental qubits and consider
\begin{align}
 H_{ZZ}^{(2E)}
 &=J_AZ_AZ_{E_A}+J_BZ_BZ_{E_B},
 \label{eq:two-ZZ}\\
 H_{XY}^{(2E)}
 &=g_A\left(\sigma_A^+\sigma_{E_A}^-+\sigma_A^-\sigma_{E_A}^+\right)\nonumber\\
 &\quad+g_B\left(\sigma_B^+\sigma_{E_B}^-+\sigma_B^-\sigma_{E_B}^+\right).
 \label{eq:two-XY}
\end{align}
For the product initial memory states considered below, the two interactions act independently on disjoint pairs, so the process factorizes spatially as $\Pcal_\theta^{AB}=\Pcal_\theta^A\otimes\Pcal_\theta^B$.  This model is not intended as an example of common-mode spatial crosstalk; rather, it isolates the increase in temporal-memory capacity when $d_E$ grows from $2$ to $4$.

With $\rho_E=\mathbb I_{E_AE_B}/4$, Eq.~\eqref{eq:two-ZZ} transmits two independent classical bits.  At $\theta=\pi/4$,
\begin{equation}
 \bm{\rmN}_{D}^{\mathrm{Choi}}(\Pcal_{\pi/4})=2,
\end{equation}
as seen from the solid orange curve of Fig.\subref{fig:classical-ZZ}. Equation~\eqref{eq:application-choi-lower} and the classical $d_E=4$ upper bound then imply
\begin{equation}
 \bm{\rmN}_{D}^{\mathrm{Choi}}(\Pcal_{\pi/4})
 =\bm{\rmN}_{D}(\Pcal_{\pi/4})=2.
 \label{eq:two-ZZ-saturation}
\end{equation}

Likewise, with $\rho_E=\ketbra{00}{00}_{E_AE_B}$, Eq.~\eqref{eq:two-XY} generates two independent $i$SWAPs at $\theta=\pi/2$.  Each environmental qubit transports one qubit of quantum memory between the two temporal blocks, giving
\begin{equation}
 \bm{\rmN}_{D}^{\mathrm{Choi}}(\Pcal_{\pi/2})=4, 
\end{equation}
as shown by the solid orange curve in Fig.~\subref{fig:quantum-XY}. Together with the quantum $d_E=4$ bound,
\begin{equation}
 \bm{\rmN}_{D}^{\mathrm{Choi}}(\Pcal_{\pi/2})
 =\bm{\rmN}_{D}(\Pcal_{\pi/2})=4.
 \label{eq:two-XY-saturation}
\end{equation}
This value also reaches the universal system ceiling in Eq.~\eqref{eq:application-system-ceiling}.  Consequently, the Choi lower bound alone certifies the equality in Eq.~\eqref{eq:two-XY-saturation}, independently of any assumption about the environmental realization.  The known $d_E=4$ realization additionally shows that the environment bound is tight.  A single environmental qubit restricts the temporal correlations to at most two bits, whereas a two-qubit environment can reach the four-bit system ceiling. An alternative single-qubit-memory construction attaining the two-bit ceiling is given in App.~\ref{app:application-tightness}.

More generally, an observed value $W_D=\bm{\rmN}_{D}^{\mathrm{Choi}}(\Pcal)$ satisfies
\begin{equation}
 W_D\le\bm{\rmN}_{D}(\Pcal)
 \le\min\{2\log d_E,2\log d_S\}.
 \label{eq:application-combined-witness-bound}
\end{equation}
Hence $W_D>2$ rules out every single-qubit quantum-memory realization, while $W_D>1$ rules out every classical single-qubit-memory realization.  More generally, an arbitrary quantum-memory realization requires
\begin{equation}
 d_E\ge\left\lceil2^{W_D/2}\right\rceil,
 \label{eq:application-quantum-dimension-witness}
\end{equation}
whereas a classical-memory realization requires
\begin{equation}
 d_E\ge\left\lceil2^{W_D}\right\rceil.
 \label{eq:application-classical-dimension-witness}
\end{equation}
Whenever the Choi lower bound reaches one of the dimension-dependent upper bounds, as in Eqs.~\eqref{eq:two-ZZ-saturation} and \eqref{eq:two-XY-saturation}, it determines the operational quantifier exactly. In particular, $W_D=4$ reaches the universal $d_S=4$ system ceiling and hence certifies the exact value without any prior knowledge of $d_E$.  Conversely, a certified value $W_D>4$ cannot be explained by increasing the environment dimension.  It would instead signal a failure of the assumed two-step, two-qubit process model, for example because of leakage outside the computational subspace, an incorrect temporal coarse graining, or errors in the reconstruction of $W_D$.

\begin{figure*}[t]
 \centering
 \includegraphics[width=0.9\textwidth]{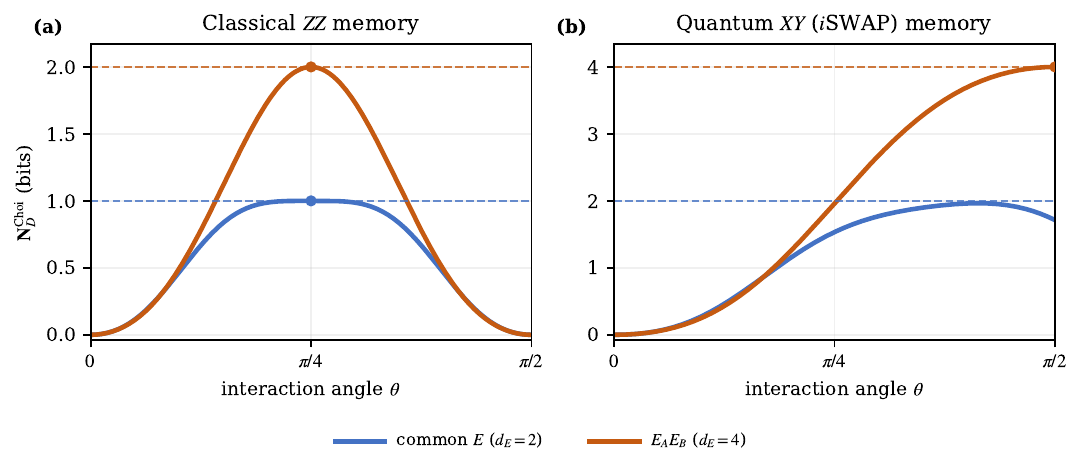}
 \subfloat[]{\label{fig:classical-ZZ}}
 \subfloat[]{\label{fig:quantum-XY}}
 \caption{Memory-dimension signatures for a two-step process on two system qubits.  (a) Relative-entropy Choi temporal correlations for classical $ZZ$ memory.  The common $d_E=2$ memory reaches its one-bit ceiling, while the two independent memories reach the $d_E=4$ ceiling of two bits.  (b) Corresponding results for quantum $XY$ exchange.  The common memory crosses the classical ceiling and approaches the two-bit quantum ceiling.  Two independent $i$SWAPs reach four bits at $\theta=\pi/2$, saturating both the $d_E=4$ environment bound and the universal $d_S=4$ system ceiling.  Dashed lines show the classical- and quantum-memory bounds; at four bits, the $d_E=4$ quantum-memory bound coincides with the system ceiling.}
 \label{fig:memory-dimension}
\end{figure*}

\subsection{Interpolation between spatial, temporal, and total correlations}
\label{subsec:application-interpolation}

The common-environment models generally generate both spatial and temporal correlations.  They therefore also provide a direct illustration of the relations between the three Choi correlation functionals.  Write
\begin{align}
 N_D(\theta)&=\bm{\rmN}_{D}^{\mathrm{Choi}}(\Pcal_\theta),&
 R_D(\theta)&=\bm{\rmR}_{D}^{\mathrm{Choi}}(\Pcal_\theta),\nonumber\\
 T_D(\theta)&=\bm{\rmT}_{D}^{\mathrm{Choi}}(\Pcal_\theta).&&
 \label{eq:application-NRT-abbreviations}
\end{align}
For the present case, $n=m=2$, the relative-entropy Choi interpolation bounds give
\begin{equation}
 L_D(\theta)\le T_D(\theta)\le U_D(\theta),
 \label{eq:application-interpolation}
\end{equation}
where
\begin{align}
 L_D(\theta)&=\max\{N_D(\theta),R_D(\theta)\},
 \label{eq:application-lower-interpolation}\\
 U_D(\theta)&=\min\big\{
 N_D(\theta)+2R_D(\theta),\nonumber\\
 &\hspace{31mm}R_D(\theta)+2N_D(\theta),8\big\}.
 \label{eq:application-upper-interpolation}
\end{align}
The last entry in Eq.~\eqref{eq:application-upper-interpolation} is the universal total-correlation system ceiling $2n\log d_S=8$.

Figure~\ref{fig:application-interpolation} displays the resulting interval for the common-$E$ $ZZ$ and $XY$ models.  The total Choi correlations remain between the lower and upper bounds for every interaction angle.  The bounds need not be saturated: their role is to constrain the total departure from fully uncorrelated noise using its separately quantified temporal and spatial components.

For the two-environment models of Sec.~\ref{subsec:two-environments}, the process factorizes spatially as $\Pcal_\theta^{AB}=\Pcal_\theta^A\otimes\Pcal_\theta^B$.  Consequently, $R_D(\theta)=0$, and Eq.~\eqref{eq:application-interpolation} collapses to the exact identity
\begin{equation}
 \bm{\rmT}_{D}^{\mathrm{Choi}}(\Pcal_\theta)
 =\bm{\rmN}_{D}^{\mathrm{Choi}}(\Pcal_\theta).
 \label{eq:application-factorized-identity}
\end{equation}
No separate interpolation plot is therefore needed for those models.

\begin{figure*}[t]
 \centering
 \includegraphics[width=0.94\textwidth]{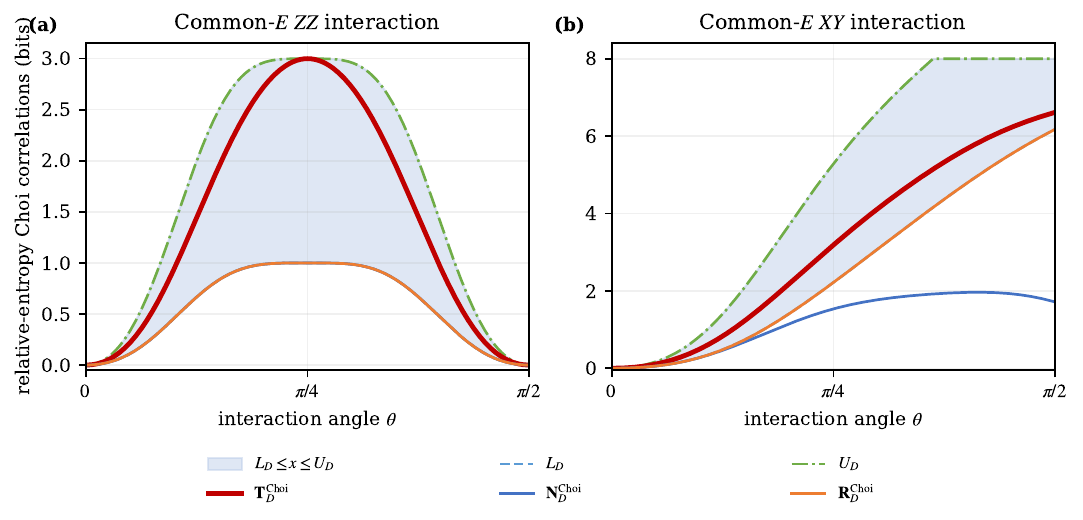}
 \caption{Relative-entropy Choi interpolation bounds for the two-step, two-qubit common-environment models.  The thick red curve shows the total Choi correlations $T_D$, while the shaded interval is bounded below by $L_D=\max\{N_D,R_D\}$ and above by $U_D=\min\{N_D+2R_D,R_D+2N_D,8\}$.  The temporal and spatial Choi correlations are also shown separately.  Panels (a) and (b) correspond to the common-$E$ $ZZ$ and $XY$ interactions, respectively. In panel (a) we have $N_D=R_D=L_D$. The total correlations lie inside the predicted interval for every interaction angle.}
 \label{fig:application-interpolation}
\end{figure*}

\section{Conclusion}\label{sec:conclusion}

We have formulated temporal, spatial, and total spatiotemporal correlations in general quantum processes through a common operational principle: optimal distinguishability from the corresponding sets of uncorrelated processes. The resulting quantities are resource-theoretically motivated and satisfy data processing under all free-set-preserving superprocesses. This operational layer complements existing methods for representing, reconstructing, and modeling spatiotemporally correlated noise by assigning the correlations a physically consistent operational magnitude.

The analytical bounds connect this magnitude to the memory responsible for temporal information flow. In particular, the distinction between classical and quantum memory yields different dimension-dependent ceilings, while the system dimension imposes a universal upper limit. When combined with any certified lower estimate of the operational correlation, these ceilings become witnesses of minimum memory dimension, bounded-memory nonclassicality, or inadequacy of the assumed process model. Importantly, the certification logic does not require the restricted estimate itself to attain the fully optimized distinguishability. The effective \(ZZ\)- and \(XY\)-interaction models illustrate how these results operate in regimes motivated by superconducting hardware. They provide examples that saturate the classical-memory, quantum-memory, and system-dimension bounds and show how an accessible Choi-state calculation can be converted into a statement about the memory compatible with the observed process.

Together, the results establish a route from multitime process data to operational correlation quantification and physically interpretable memory constraints. An important next step is to combine these tools with scalable process-learning and restricted-probing methods, allowing the relevant bounds to be evaluated without reconstructing exponentially large process tensors. Such developments could support correlation-aware characterization, control, and noise-management strategies for increasingly complex quantum devices.

\begin{acknowledgments}
We thank Marco Túlio Quintino, Ernesto Galvão, and Frederico Brito for useful discussions. OpenAI's ChatGPT (GPT-5.6 Sol) was used to assist with drafting and revising portions of the manuscript, exploring and checking intermediate mathematical arguments, and generating code for the numerical examples. The authors directed its use through task-specific prompts and independently verified all mathematical statements, plots, and manuscript text. All scientific decisions and conclusions were made by the authors, who take full responsibility for the content of this work. This work was supported by the São Paulo Research Foundation (FAPESP) under Grant No.~2022/00993-9 and Grant No.~2023/04625-7 and by the National Council for Scientific and Technological Development (CNPq) under Grant No.~305766/2026-0.
\end{acknowledgments}

\section*{Data Availability}

All plots presented in this work were generated directly from analytical expressions specified in the text. The code used to evaluate these expressions and generate the figures is available from the corresponding author upon reasonable request.

\appendix

\section{Relative-entropy counterexample to Choi monotonicity}
\label{app:counterexample}

We adapt the temporal counterexample of Ref.~\cite{zambon2024process} to a bipartite spatial channel. Let $\bar A,\bar C$ denote two qubit inputs and $B,D$ the outputs. We write \(\Id_X=\mathbb I_X/2\) for the maximally mixed qubit state and \(\Phi_{AD}=\ketbra{\Phi}{\Phi}_{AD}\), where
\(\ket{\Phi}_{AD}=(\ket{00}_{AD}+\ket{11}_{AD})/\sqrt{2}\). Define the CPTP channel
\begin{align}
 \Ecal_{\bar A\bar C\to BD}(X)
 =\ketbra{0}{0}_B\otimes\big[
 &\Ical_{\bar A\to D}
   (\bra0_{\bar C}X\ket0_{\bar C})\nonumber\\
 &+\tr_{\bar A}(\bra1_{\bar C}X\ket1_{\bar C})\Id_D
 \big].
 \label{eq:counterexample-channel}
\end{align}
Its normalized Choi state, with $A,C$ the input references, is
\begin{equation}
 \Upsilon^\Ecal
 =\ketbra{0}{0}_B\otimes\frac12\left[
 \ketbra{0}{0}_C\otimes\Phi_{AD}
 +\ketbra{1}{1}_C\otimes\Id_A\otimes\Id_D
 \right].
 \label{eq:counterexample-choi}
\end{equation}
We evaluate spatial total correlation across $AB:CD$. The relevant von Neumann entropies are
\begin{align}
 H(AB)&=1,&H(CD)&=2,&H(ABCD)&=2,
 \end{align}
and therefore
\begin{align}
    \CD^{AB:CD}(\Ecal)&=I(AB:CD)_{\Upsilon^{\Ecal}}\\
    &= H(AB)+H(CD)-H(ABCD)\\
    &=1.
\end{align}

Now act locally on the second party by precomposing the $\bar C$ input with the replacer
\begin{equation}
 \Rcal_0(Y)=\tr(Y)\ketbra{0}{0}_{\bar C}.
 \label{eq:replacer-zero}
\end{equation}
This is a deterministic local superchannel and maps every product channel to a product channel. The transformed normalized Choi state is
\begin{equation}
 \Upsilon^{\Ecal\circ(\Ical_{\bar A}\otimes\Rcal_0)}
 =\Phi_{AD}\otimes\ketbra{0}{0}_B\otimes\Id_C .
 \label{eq:counterexample-after-choi}
\end{equation}
Now
\begin{align}
 H(AB)&=1,&H(CD)&=2,&H(ABCD)&=1,
 \end{align}
so
\begin{equation}
 \CD^{AB:CD}\!\left[
 \Ecal\circ(\Ical_{\bar A}\otimes\Rcal_0)\right]=2>1.
 \label{eq:counterexample-after}
\end{equation}
Thus relative-entropy Choi correlation increases under a local correlation-nongenerating superchannel.  In the normalization of Ref.~\cite{rivas2015quantifying}, the same increase is $1/4\to1/2$.

\section{Proofs of correlation-quantifier relations}
\label{app:quantifier-proofs}

\subsection{Monotonicity}
\label{app:monotonicity}

\begin{proof}[Proof of Theorem~\ref{thm:monotonicity}]
We start from the definition and then use two inequalities: the first one is a consequence of $\bm{Z}(\Xbb)\subseteq\Xbb$ and the second one is the monotonicity of generalized divergences under superprocesses (Theorem 1 of Ref.~\cite{zambon2024process}), yielding
\begin{align}
    \bm{\rmG}_{S}^{\Xbb}\!\left[\bm{Z}(\Pcal)\right]&= \inf_{\Qcal\in\Xbb}G_{S}\!\left[\bm{Z}(\Pcal),\Qcal\right]\\
    &\le \inf_{\Qcal\in\bm{Z}(\Xbb)}G_{S}\!\left[\bm{Z}(\Pcal),\Qcal\right]\\
    &= \inf_{\Qcal\in\Xbb}G_{S}\!\left[\bm{Z}(\Pcal),\bm{Z}(\Qcal)\right]\\
    &\le \inf_{\Qcal\in\Xbb} G_{S}(\Pcal,\Qcal)\\
    &= \bm{\rmG}_{S}^{\Xbb}(\Pcal),
\end{align}
as stated in Thm.~\ref{thm:monotonicity}.
\end{proof}

\subsection{Correlation hierarchy}\label{app:quantifier-relations}

\begin{proof}[Proof of Proposition~\ref{prop:hierarchy}]
Write the operational process divergence as
\begin{equation}
 G_S(\Pcal,\Qcal)
 =\sup_{\Scal\in\Sbb}
 S\!\left[\Pcal(\Scal),\Qcal(\Scal)\right].
 \label{eq:process-divergence-app}
\end{equation}
By definition,
\begin{align}
 \NS&=\GS^{\Mbb}=\inf_{\Qcal\in\Mbb}G_S(\Pcal,\Qcal),\\
 \RS&=\GS^{\Lbb}=\inf_{\Qcal\in\Lbb}G_S(\Pcal,\Qcal),\\
 \TS&=\GS^{\Ubb}=\inf_{\Qcal\in\Ubb}G_S(\Pcal,\Qcal).
 \label{eq:three-operational-defs-app}
\end{align}
Since $\Ubb=\Mbb\cap\Lbb$ is a subset of each larger free set, minimization over $\Ubb$ cannot decrease either of the first two minima. This proves Eq.~\eqref{eq:operational-lower-hierarchy}.  The same argument applies to the corresponding Choi free sets.  Restricting the maximum in Eq.~\eqref{eq:process-divergence-app} to any subclass of probes cannot increase it; minimizing afterwards proves Eq.~\eqref{eq:choi-operational-order} and its restricted-probe generalization.
\end{proof}

\subsection{Interpolation bounds}

\subsubsection{Trace-distance interpolation}\label{app:trace-interpolation}

\begin{proof}[Proof of Proposition~\ref{prop:trace-interpolation}]
For a process $\Pcal\in\Pbb_n^m$ let $\Mcal\in\Mbb_n^m$ be the closest Markovian process and $\Lcal\in\Lbb_n^m$ be the closest local process, with
\begin{equation}
 \Mcal=\bigotimes_{j=1}^{n}\Mcal_j,
 \qquad
 \Lcal=\bigotimes_{k=1}^{m}\Lcal_k,
 \label{eq:epsilon-free-approximants}
\end{equation}
such that
\begin{equation}
 G_T(\Pcal,\Mcal)=\bm{\rmN}_{T}(\Pcal),
 \qquad
 G_T(\Pcal,\Lcal)=\bm{\rmR}_{T}(\Pcal).
 \label{eq:epsilon-distances}
\end{equation}
Here $\Mcal_j$ denotes a complete, possibly spatially correlated, one-step block, whereas $\Lcal_k$ may contain arbitrary temporal memory on subsystem $k$. These minimizers exist because the free sets are compact in finite dimensions and the process trace distance is continuous.

Let $\mathfrak R_j$ be a deterministic causal reduction that retains only temporal block $j$ and closes all other slots with fixed deterministic local operations.  Then
\begin{equation}
 \mathfrak R_j(\Mcal)=\Mcal_j,
 \qquad
 \widetilde{\Lcal}_j
 \Def\mathfrak R_j(\Lcal)
 =\bigotimes_{k=1}^{m}\Lcal_{jk}
 \label{eq:time-reductions}
\end{equation}
is spatially local.  Contractivity and the triangle inequality give
\begin{align}
 G_T(\Mcal_j,\widetilde{\Lcal}_j)
 &\le G_T(\Mcal,\Lcal)\\
 &\le G_T(\Mcal,\Pcal)
      +G_T(\Pcal,\Lcal)\\
 &\le\bm{\rmN}_{T}(\Pcal)+\bm{\rmR}_{T}(\Pcal).
 \label{eq:block-distance-bound}
\end{align}

The process trace distance is tensor-subadditive.  To see this, first note that tensoring both arguments with the same deterministic process cannot increase their distance: every joint probing strategy induces a valid strategy on the unfixed factor after the deterministic factor is inserted. Applying this stability to the telescoping identity for two tensor products and then the triangle inequality gives
\begin{equation}
 G_T\!\left(\bigotimes_j\Pcal_j,\bigotimes_j\Qcal_j\right)
 \le\sum_jG_T(\Pcal_j,\Qcal_j).
 \label{eq:process-tensor-subadditivity}
\end{equation}

Now define
\begin{equation}
 \Ucal
 =\bigotimes_{j=1}^{n}\widetilde{\Lcal}_j\in\Ubb.
 \label{eq:constructed-L}
\end{equation}
Equations~\eqref{eq:block-distance-bound} and \eqref{eq:process-tensor-subadditivity} imply
\begin{equation}
 G_T(\Mcal,\Ucal)
 \le n[\bm{\rmN}_{T}(\Pcal)+\bm{\rmR}_{T}(\Pcal)].
 \label{eq:M-to-L-bound}
\end{equation}
Therefore
\begin{align}
 \bm{\rmT}_{T}(\Pcal)
 &\le G_T(\Pcal,\Ucal)\\
 &\le\bm{\rmN}_{T}(\Pcal)+n\qty[\bm{\rmN}_{T}(\Pcal)+\bm{\rmR}_{T}(\Pcal)].
\end{align}
This proves Eq.~\eqref{eq:trace-interpolation-n}. Reversing the roles of the temporal and spatial factorizations gives Eq.~\eqref{eq:trace-interpolation-m}.
\end{proof}

\subsubsection{Relative-entropy Choi interpolation}\label{app:choi-interpolation}

\begin{proof}[Proof of Proposition~\ref{prop:choi-interpolation}]
Let $A_{jk}$ denote the normalized-Choi Hilbert-space factor associated with time $j$ and subsystem $k$, and group the factors as
\begin{equation}
 A_{j:}=\bigotimes_{k=1}^{m}A_{jk},\qquad
 A_{:k}=\bigotimes_{j=1}^{n}A_{jk}.
\end{equation}
For the product Choi partitions used in the definitions, the closest product state in relative entropy is the product of the corresponding marginals. Hence
\begin{align}
 \bm{\rmT}_{D}^{\mathrm{Choi}}  &=\sum_{j,k}H(\Upsilon_{A_{jk}})-H(\Upsilon),
 \label{eq:T-choi-entropy}\\
 \bm{\rmN}_{D}^{\mathrm{Choi}}  &=\sum_jH(\Upsilon_{A_{j:}})-H(\Upsilon),
 \label{eq:N-choi-entropy}\\
 \bm{\rmR}_{D}^{\mathrm{Choi}}  &=\sum_kH(\Upsilon_{A_{:k}})-H(\Upsilon).
 \label{eq:R-choi-entropy}
\end{align}
It follows exactly that
\begin{align}
 \bm{\rmT}_{D}^{\mathrm{Choi}}-\bm{\rmN}_{D}^{\mathrm{Choi}}  &=\sum_{j=1}^{n}
 \left[\sum_kH(\Upsilon_{A_{jk}})-H(\Upsilon_{A_{j:}})\right]\\
 &=\sum_{j=1}^{n}
 D\!\left(\Upsilon_{A_{j:}}\middle\|
 \bigotimes_k\Upsilon_{A_{jk}}\right).
 \label{eq:choi-chain-exact}
\end{align}
Furthermore,
\begin{equation}
 \bm{\rmR}_{D}^{\mathrm{Choi}}
 =D\!\left(\Upsilon\middle\|
 \bigotimes_k\Upsilon_{A_{:k}}\right).
\end{equation}
Tracing out every temporal block except $j$ maps the second argument to $\bigotimes_k\Upsilon_{A_{jk}}$.  Data processing therefore gives
\begin{equation}
 D\!\left(\Upsilon_{A_{j:}}\middle\|
 \bigotimes_k\Upsilon_{A_{jk}}\right)
 \le \bm{\rmR}_{D}^{\mathrm{Choi}}.
 \label{eq:choi-DPI-step}
\end{equation}
Summing Eq.~\eqref{eq:choi-DPI-step} over $j$ proves Eq.~\eqref{eq:choi-interpolation-n}.  Exchanging $j$ and $k$ proves Eq.~\eqref{eq:choi-interpolation-m}.

Moreover, the above bounds are tight within the set of valid combs. Consider a causal replacer process that ignores every input, samples one unbiased classical bit at the first step, stores it, and emits the same bit on every subsystem at every step. Its normalized Choi state is
\begin{equation}
 \Upsilon_{m,n}
 =\left(
 \frac12\ketbra{0}{0}_{O}^{\otimes mn}
 +\frac12\ketbra{1}{1}_{O}^{\otimes mn}
 \right)
 \otimes\bigotimes_{j,k}\Id_{i_{jk}}.
 \label{eq:perfect-classical-grid}
\end{equation}
The maximally mixed input factors are independent and cancel from every total-correlation expression. Every nonempty output marginal has entropy one bit, so
\begin{equation}
 \bm{\rmT}_{D}^{\mathrm{Choi}}=mn-1,\qquad
 \bm{\rmN}_{D}^{\mathrm{Choi}}=n-1,\qquad
 \bm{\rmR}_{D}^{\mathrm{Choi}}=m-1.
\end{equation}
Consequently,
\begin{equation}
 \bm{\rmT}_{D}^{\mathrm{Choi}}
 =\bm{\rmN}_{D}^{\mathrm{Choi}}+n\bm{\rmR}_{D}^{\mathrm{Choi}}
 =\bm{\rmR}_{D}^{\mathrm{Choi}}+m\bm{\rmN}_{D}^{\mathrm{Choi}},
\end{equation}
thus saturating the bounds.
\end{proof}

\section{Domination tools}

\subsection{State domination}\label{app:domination}

\begin{proof}[Proof of Lemma~\ref{lem:state-domination}]
Following the standard binary preparation construction \cite{wang2019resource,gour2020optimal}, for $0\le\lambda<1$ define
\begin{equation}
 \tau=\frac{\sigma-\lambda\rho}{1-\lambda}.
\end{equation}
Since $\sigma\ge\lambda\rho$ and $1-\lambda>0$, $\tau\ge0$. Moreover, $\tr[\tau]=1$, implying $\tau$ is a state and
\begin{equation}
 \sigma=\lambda\rho+(1-\lambda)\tau.
\end{equation}
Define the preparation channel
\begin{equation}
 \Lambda(X)
 =\langle0|X|0\rangle\rho+\langle1|X|1\rangle\tau,
\end{equation}
and the binary states
\begin{equation}
 p=\ketbra{0}{0},\qquad
 q_\lambda=\lambda\ketbra{0}{0}+(1-\lambda)\ketbra{1}{1}.
\end{equation}
Since $\Lambda(p)=\rho$ and $\Lambda(q_\lambda)=\sigma$, data processing gives
\begin{equation}
 S(\rho,\sigma)
 =S\!\left[\Lambda(p),\Lambda(q_\lambda)\right]
 \le S(p,q_\lambda)
 =b_S(\lambda),
\end{equation}
which proves the lemma for $0\le\lambda<1$. If $\lambda=1$, normalized-state domination implies $\rho=\sigma$; applying data processing to any preparation channel that maps $p$ to $\rho$ gives $S(\rho,\rho)\le S(p,p)=b_S(1)$.

It remains to verify the stated monotonicity. For $0\le\lambda'\le\lambda<1$, let a binary classical channel fix $\ketbra{0}{0}$ and map $\ketbra{1}{1}$ to $r\ketbra{0}{0}+(1-r)\ketbra{1}{1}$, where
\begin{equation}
 r=\frac{\lambda-\lambda'}{1-\lambda'}.
\end{equation}
This channel maps $q_{\lambda'}$ to $q_\lambda$ while fixing $p$. Data processing therefore gives $b_S(\lambda)\le b_S(\lambda')$. The case $\lambda=1$ follows by using the constant channel with output $p$, so $b_S$ is nonincreasing on $[0,1]$.

For completeness, the two calibrations used in the main text follow directly. For relative entropy and $\lambda>0$,
\begin{align}
 b_D(\lambda)
 &=D(p\Vert q_\lambda)\nonumber\\
 &=-\bra{0}\log q_\lambda\ket{0}\nonumber\\
 &=-\log\lambda.
 \label{eq:bD-derivation}
\end{align}
At $\lambda=0$, the support of $p$ is not contained in that of $q_0$, so $b_D(0)=+\infty$. For trace distance,
\begin{align}
 b_T(\lambda)
 &=\frac12\norm{p-q_\lambda}_1\nonumber\\
 &=\frac{1-\lambda}{2}
 \norm{\ketbra{0}{0}-\ketbra{1}{1}}_1\nonumber\\
 &=1-\lambda.
 \label{eq:bT-derivation}
\end{align}
\end{proof}

\subsection{Free-candidate domination}
\label{app:free-candidate-domination}

The following process-level consequence of Lemma~\ref{lem:state-domination} is the form used in the dimension-bound proofs.

\begin{corollary}[Free-candidate domination]
\label{cor:free-candidate-domination}
Let $\Pcal,\Qcal\in\Pbb_n$ have normalized Choi operators satisfying
\begin{equation}
 \lambda\Upsilon^{\Pcal}\le\Upsilon^{\Qcal},
 \qquad \lambda\in[0,1].
 \label{eq:free-candidate-choi-domination}
\end{equation}
Then
\begin{equation}
 G_S(\Pcal,\Qcal)\le b_S(\lambda).
 \label{eq:free-candidate-process-bound}
\end{equation}
Consequently, if $\Qcal\in\Xbb\subseteq\Pbb_n$, then
\begin{equation}
 \bm{\rmG}_{S}^{\Xbb}(\Pcal)\le b_S(\lambda).
 \label{eq:free-candidate-correlation-bound}
\end{equation}
\end{corollary}

\begin{proof}
For any deterministic compatible probing comb $\Scal$, contracting Eq.~\eqref{eq:free-candidate-choi-domination} with $\Upsilon^{\Scal}$ preserves the CP ordering. Indeed, $\Upsilon^{\Qcal}-\lambda\Upsilon^{\Pcal}$ is the Choi operator of a completely positive comb, and its contraction with a compatible completely positive comb is positive. Hence
\begin{equation}
 \lambda\Pcal(\Scal)\le\Qcal(\Scal).
 \label{eq:free-candidate-output-domination}
\end{equation}
Both outputs are normalized states, so Lemma~\ref{lem:state-domination} gives
\begin{equation}
 S\!\left[\Pcal(\Scal),\Qcal(\Scal)\right]\le b_S(\lambda)
 \qquad\forall\Scal\in\Sbb_n.
\end{equation}
Maximizing over $\Scal$ proves Eq.~\eqref{eq:free-candidate-process-bound}. If $\Qcal\in\Xbb$, using it as a candidate in the minimization defining $\bm{\rmG}_{S}^{\Xbb}$ proves Eq.~\eqref{eq:free-candidate-correlation-bound}.
\end{proof}

\subsection{Marginal domination}

\begin{lemma}[Marginal domination]
\label{lem:marginal-domination}
Let $X_{AB}\ge0$, with $X_B=\tr_A X_{AB}$, and set $d_A=\dim\mathcal H_A$. Then
\begin{equation}
 X_{AB}\le d_A\,\mathbb I_A\otimes X_B.
 \label{eq:marginal-domination}
\end{equation}
\end{lemma}

\begin{proof}
Consider the linear map
\begin{equation}
 \Gamma_A(Z)=d_A\tr(Z)\mathbb I_A-Z.
\end{equation}
For the unnormalized vector $\ket{\Omega}=\sum_{j=1}^{d_A}\ket{jj}$, its Choi operator is
\begin{equation}
 J(\Gamma_A)
 =d_A\mathbb I_{AA'}-\ket{\Omega}\!\bra{\Omega}\ge0,
\end{equation}
because $\ket{\Omega}\!\bra{\Omega}$ has a single nonzero eigenvalue $d_A$. Thus $\Gamma_A$ is completely positive, and
\begin{equation}
 (\Gamma_A\otimes\Ical_B)(X_{AB})
 =d_A\mathbb I_A\otimes X_B-X_{AB}\ge0,
\end{equation}
which proves Eq.~\eqref{eq:marginal-domination}.
\end{proof}

\section{Dimension bounds and certification}

\subsection{Quantum and classical memory bounds}\label{app:memory-bounds}

\begin{proof}[Proof of Theorems~\ref{thm:memory-bound} and \ref{thm:classical-memory-bound}]
Open the $n-1$ inter-step memory wires in the realization of Eq.~\eqref{eq:memory-channel-realization}, and denote the resulting completely positive comb map by $\Ncal$. For $\kappa\in\{1,2\}$, define
\begin{equation}
 \Acal_E^{(\kappa)}
 =\bigotimes_{j=1}^{n-1}\Acal_{E_j}^{(\kappa)},
 \qquad
 \Acal_{E_j}^{(2)}=\Ical_{E_j},
 \quad
 \Acal_{E_j}^{(1)}=\Zcal_{E_j}.
 \label{eq:kappa-memory-channels}
\end{equation}
The quantum-memory realization is $\Pcal=\Ncal\!\left[\Acal_E^{(2)}\right]$. For a classical-memory realization, Definition~\ref{def:classical-memory} allows the identity channel on every open memory wire to be replaced by complete dephasing without changing the process, so $\Pcal=\Ncal\!\left[\Acal_E^{(1)}\right]$.

Let
\begin{equation}
 \Delta_E=\bigotimes_{j=1}^{n-1}\Delta_{E_j},
 \qquad
 \Delta_{E_j}(X)=\tr(X)\Id_{E_j},
\end{equation}
and define $\Pcal^\Delta=\Ncal(\Delta_E)$. Each depolarizing channel discards the incoming memory and prepares a fixed state, thereby breaking every temporal link; hence $\Pcal^\Delta\in\Mbb_n$.

For a link of dimension $d_j$, the normalized Choi operators are
\begin{align}
 \Upsilon^{\Delta_{E_j}}
 &=\frac{\mathbb I_{E_j^{\rm o}E_j^{\rm i}}}{d_j^2},\\
 \Upsilon^{\Ical_{E_j}}
 &=\ketbra{\Phi_j}{\Phi_j},\\
 \Upsilon^{\Zcal_{E_j}}
 &=\frac{1}{d_j}\sum_{a=1}^{d_j}\ketbra{aa}{aa},
\end{align}
where $\ket{\Phi_j}=d_j^{-1/2}\sum_a\ket{aa}$. Therefore
\begin{equation}
 \Upsilon^{\Acal_{E_j}^{(\kappa)}}
 \le d_j^\kappa\Upsilon^{\Delta_{E_j}},
 \qquad \kappa\in\{1,2\}.
\end{equation}
Tensoring these inequalities gives
\begin{equation}
 \lambda_\kappa\Upsilon^{\Acal_E^{(\kappa)}}
 \le\Upsilon^{\Delta_E},
 \qquad
 \lambda_\kappa=\prod_{j=1}^{n-1}d_j^{-\kappa}.
 \label{eq:kappa-memory-domination}
\end{equation}
Since $\Ncal$ is completely positive, it preserves this order, and thus
\begin{equation}
 \lambda_\kappa\Upsilon^{\Pcal}
 \le\Upsilon^{\Pcal^\Delta}.
 \label{eq:process-memory-domination}
\end{equation}
Corollary~\ref{cor:free-candidate-domination}, with $\Pcal^\Delta\in\Mbb_n$, now yields
\begin{equation}
 \NS(\Pcal)\le b_S(\lambda_\kappa).
\end{equation}
Choosing $\kappa=2$ proves Eq.~\eqref{eq:memory-profile-bound}, whereas $\kappa=1$ proves Eq.~\eqref{eq:classical-profile-bound}. If $d_j\le d_E$ at every cut, then $\lambda_\kappa\ge d_E^{-\kappa(n-1)}$; because $b_S$ is nonincreasing, Eqs.~\eqref{eq:environment-bound} and \eqref{eq:classical-environment-bound} follow. The Choi bounds follow by restricting the probing optimization to the Choi-implementing comb.
\end{proof}

\subsection{Universal system ceilings}
\label{app:system-ceilings}

\begin{proof}[Proof of Theorem~\ref{thm:system-ceilings}]
Let \(\Upsilon^{(k)}\) be the normalized Choi state of the reduced process containing its first \(k\) steps.  Normalized comb causality gives
\begin{equation}
 \tr_{o_k}\Upsilon^{(k)}
 =\Id_{i_k}\otimes\Upsilon^{(k-1)},
 \qquad
 \Id_{i_k}=\frac{\mathbb I_{i_k}}{d_S}.
 \label{eq:normalized-comb-causality}
\end{equation}
Applying Lemma~\ref{lem:marginal-domination} with \(A=o_k\) yields
\begin{align}
 \Upsilon^{(k)}
 &\le d_S\mathbb I_{o_k}\otimes
       \tr_{o_k}\Upsilon^{(k)}\nonumber\\
 &=d_S\mathbb I_{o_k}\otimes\Id_{i_k}
       \otimes\Upsilon^{(k-1)}\nonumber\\
 &=d_S^2\Id_{o_k}\otimes\Id_{i_k}
       \otimes\Upsilon^{(k-1)},
 \label{eq:recursive-system-domination}
\end{align}
where \(\Id_{o_k}=\mathbb I_{o_k}/d_S\).

Iterating Eq.~\eqref{eq:recursive-system-domination} from \(k=n\) down to \(k=2\) gives
\begin{equation}
 d_S^{-2(n-1)}\Upsilon^{\Pcal}
 \le\Upsilon^{\Qcal_{\rm M}},
 \label{eq:Markovian-system-domination}
\end{equation}
with
\begin{equation}
 \Upsilon^{\Qcal_{\rm M}}
 =\Upsilon^{(1)}\otimes
   \bigotimes_{k=2}^{n}(\Id_{o_k}\otimes\Id_{i_k}).
 \label{eq:Markovian-system-candidate}
\end{equation}
This is the normalized Choi state of a Markovian process: the first block is the original reduced one-step process and all later blocks are maximally mixed replacers.  Hence \(\Qcal_{\rm M}\in\Mbb\).

Applying Corollary~\ref{cor:free-candidate-domination} to Eq.~\eqref{eq:Markovian-system-domination}, with $\Qcal_{\rm M}\in\Mbb$, directly yields Eq.~\eqref{eq:N-system-ceiling}.

Applying Eq.~\eqref{eq:recursive-system-domination} also to \(k=1\), with \(\Upsilon^{(0)}=1\), yields
\begin{equation}
 d_S^{-2n}\Upsilon^{\Pcal}
 \le\Upsilon^{\Qcal_{\rm U}},
 \qquad
 \Upsilon^{\Qcal_{\rm U}}
 =\bigotimes_{k=1}^{n}(\Id_{o_k}\otimes\Id_{i_k}).
 \label{eq:uncorrelated-system-domination}
\end{equation}
The process \(\Qcal_{\rm U}\) is a tensor product over time of completely depolarizing channels.  Moreover, if \(\mathcal H_S=\bigotimes_{a=1}^{m}\mathcal H_{X_a}\), then
\begin{equation}
 \Delta_S=\bigotimes_{a=1}^{m}\Delta_{X_a},
\end{equation}
so every temporal block is also spatially local. Consequently, $\Qcal_{\rm U}\in\Ubb=\Mbb\cap\Lbb$. Applying Corollary~\ref{cor:free-candidate-domination} to Eq.~\eqref{eq:uncorrelated-system-domination}, first with $\Xbb=\Lbb$ and then with $\Xbb=\Ubb$, proves Eqs.~\eqref{eq:R-system-ceiling} and \eqref{eq:T-system-ceiling}.

The Choi bounds follow from Eq.~\eqref{eq:choi-operational-order}. Finally, \(\ES^{\Ybb}\le\TS\) follows by restricting the maximization over probes from \(\Sbb\) to \(\Ybb\), proving Eq.~\eqref{eq:E-system-ceiling}. Substitution of \(b_D(\lambda)=-\log\lambda\) and \(b_T(\lambda)=1-\lambda\) proves Eqs.~\eqref{eq:ND-system-ceiling}--\eqref{eq:TT-system-ceiling}.
\end{proof}

\subsection{Certification statements}
\label{app:certification}

\begin{proof}[Proof of Corollary~\ref{cor:certification}]
For arbitrary quantum memory, the finite-memory theorem and Eq.~\eqref{eq:certified-lower-bounds} give
\begin{equation}
 W_S^{\rm N}\le b_S\!\left[d_E^{-2(n-1)}\right].
\end{equation}
By the definition of \(b_S^{\leftarrow}\), this implies
\begin{equation}
 d_E^{-2(n-1)}
 \le b_S^{\leftarrow}(W_S^{\rm N}),
\end{equation}
which proves Eq.~\eqref{eq:quantum-dimension-witness}.  For classical memory, the refined memory bound instead gives
\begin{equation}
 W_S^{\rm N}\le b_S\!\left[d_E^{-(n-1)}\right],
\end{equation}
and the identical inversion proves Eq.~\eqref{eq:classical-dimension-witness}.  Applying the same argument to Eqs.~\eqref{eq:N-system-ceiling}--\eqref{eq:T-system-ceiling} proves Eqs.~\eqref{eq:system-N-witness} and \eqref{eq:system-RT-witness}.

For relative entropy, \(b_D^{\leftarrow}(w)=2^{-w}\); for trace distance, \(b_T^{\leftarrow}(w)=1-w\) for \(0\le w<1\).  This gives Eqs.~\eqref{eq:D-quantum-dimension-witness}-- \eqref{eq:trace-classical-dimension-witness}.  Finally, under the prior information \(d_E\le d_{\max}\), the classical and quantum memory ceilings are, respectively, \(b_S\!\left[d_{\max}^{-(n-1)}\right]\) and \(b_S\!\left[d_{\max}^{-2(n-1)}\right]\).  Any certified lower bound above either ceiling contradicts the corresponding realization, proving Eqs.~\eqref{eq:general-nonclassicality-witness} and \eqref{eq:general-quantum-memory-exclusion}.
\end{proof}

\section{Superconducting-qubit model calculations}

\subsection{Choi construction and numerical evaluation}
\label{app:application-choi}

Let $A_jB_j$ be the two-qubit system input in temporal block $j$, and let $R_j=R_j^AR_j^B$ be an isomorphic reference. Define the maximally entangled Choi input for that block as $\ket{\Phi_j}=\ket\Phi_{R_j^AA_j}\ket\Phi_{R_j^BB_j}$. For a pure initial environment state $\ket\eta_E$, the normalized Choi state is obtained from
\begin{equation}
 \ket{\Psi_\theta}
 =U_{A_2B_2E}(\theta)U_{A_1B_1E}(\theta)
 \ket{\Phi_1}\ket{\Phi_2}\ket\eta_E,
 \label{eq:application-global-state}
\end{equation}
where $U(\theta)=e^{-i\theta \Tilde{H}}$, with the dimensionless Hamiltonian given by either $\Tilde{H}=H/J$ or $\Tilde{H}=H/g$ depending on the case.  Thus
\begin{equation}
 \Upsilon^{\Pcal_\theta}
 =\tr_E\ketbra{\Psi_\theta}{\Psi_\theta}.
 \label{eq:application-choi-state}
\end{equation}
For a mixed initial state, Eq.~\eqref{eq:application-choi-state} is extended linearly.  Define the four normalized-Choi cells
\begin{equation}
 C_{jA}=R_j^AA_j,
 \qquad C_{jB}=R_j^BB_j,
 \qquad j\in\{1,2\},
 \label{eq:application-Choi-cells}
\end{equation}
and group them either temporally as $T_j=C_{jA}C_{jB}$ or spatially as $V_A=C_{1A}C_{2A}$ and $V_B=C_{1B}C_{2B}$.  For relative entropy, the closest product state for each partition is the product of the corresponding marginals.  Hence
\begin{align}
 \bm{\rmN}_{D}^{\mathrm{Choi}}(\Pcal_\theta)
 &=D\!\left(\Upsilon^{\Pcal_\theta}\middle\|
 \Upsilon_{T_1}^{\Pcal_\theta}\otimes
 \Upsilon_{T_2}^{\Pcal_\theta}\right)\\
 &=H(T_1)+H(T_2)-H(T_1T_2).
 \label{eq:application-MI-calculation}
\end{align}
Similarly,
\begin{align}
 \bm{\rmR}_{D}^{\mathrm{Choi}}(\Pcal_\theta)
 &=H(V_A)+H(V_B)-H(T_1T_2),
 \label{eq:application-spatial-calculation}\\
 \bm{\rmT}_{D}^{\mathrm{Choi}}(\Pcal_\theta)
 &=\sum_{j=1}^{2}\left[H(C_{jA})+H(C_{jB})\right]-H(T_1T_2).
 \label{eq:application-total-calculation}
\end{align}
Equation~\eqref{eq:application-MI-calculation} is plotted in Fig.~\ref{fig:memory-dimension}, whereas Eqs.~\eqref{eq:application-MI-calculation}-- \eqref{eq:application-total-calculation} produce Fig.~\ref{fig:application-interpolation}.  No optimization over operational probing combs is required.

\subsection{Analytical evaluation of the $ZZ$ models}
\label{app:application-ZZ}

For the symmetric common-environment Hamiltonian,
\begin{equation}
 U_{ZZ}^{(1E)}(\theta)
 =\sum_{e=0}^{1}\ketbra e e_E\otimes
 e^{-i(-1)^e\theta(Z_A+Z_B)}.
 \label{eq:controlled-common-ZZ}
\end{equation}
Starting from a state diagonal in the $Z_E$ basis, every joint state remains block diagonal in that basis for arbitrary interventions on $AB$.  Hence inserting
\begin{equation}
 \Zcal_E(X)=\sum_{e=0}^{1}\ketbra e e X\ketbra e e
\end{equation}
at the temporal cut leaves the process invariant, proving classicality of the memory realization.

For $\rho_E=\mathbb I_E/2$, the two conditional normalized-Choi vectors have overlap
\begin{equation}
 c_{1E}(\theta)
 =\frac14\left|\tr e^{-2i\theta(Z_A+Z_B)}\right|
 =\cos^2(2\theta).
\end{equation}
Define
\begin{equation}
 F(c)=2h_2\!\left(\frac{1+c}{2}\right)
 -h_2\!\left(\frac{1+c^2}{2}\right),
 \label{eq:F-overlap}
\end{equation}
where $h_2$ is the binary entropy.  Diagonalizing the equal mixture of the two conditional Choi vectors gives
\begin{equation}
 \bm{\rmN}_{D}^{\mathrm{Choi}}\!\left[\Pcal_\theta^{(1E,ZZ)}\right]
 =F\!\left[\cos^2(2\theta)\right].
 \label{eq:common-ZZ-curve}
\end{equation}
At $\theta=\pi/4$, the conditional Choi vectors are orthogonal and Eq.~\eqref{eq:common-ZZ-curve} equals one.

The symmetry between the two time blocks and the two system qubits also gives
\begin{equation}
 \bm{\rmR}_{D}^{\mathrm{Choi}}\!\left[\Pcal_\theta^{(1E,ZZ)}\right]
 =\bm{\rmN}_{D}^{\mathrm{Choi}}\!\left[\Pcal_\theta^{(1E,ZZ)}\right].
 \label{eq:common-ZZ-N-equals-R}
\end{equation}
Writing $c(\theta)=|\cos(2\theta)|$, the corresponding four-cell total correlation is
\begin{equation}
 \bm{\rmT}_{D}^{\mathrm{Choi}}\!\left[\Pcal_\theta^{(1E,ZZ)}\right]
 =4h_2\!\left[\frac{1+c(\theta)}{2}\right]
 -h_2\!\left[\frac{1+c(\theta)^4}{2}\right].
 \label{eq:common-ZZ-total-curve}
\end{equation}
In particular, at $\theta=\pi/4$ the common classical bit produces
\begin{equation}
 \bm{\rmN}_{D}^{\mathrm{Choi}}=\bm{\rmR}_{D}^{\mathrm{Choi}}=1,
 \qquad
 \bm{\rmT}_{D}^{\mathrm{Choi}}=3.
 \label{eq:common-ZZ-NRT-special-time}
\end{equation}

For two independent environmental qubits, the pairs $AE_A$ and $BE_B$ instead have conditional-Choi overlap $|\cos(2\theta)|$.  Additivity across the two pairs therefore yields
\begin{equation}
 \bm{\rmN}_{D}^{\mathrm{Choi}}\!\left[\Pcal_\theta^{(2E,ZZ)}\right]
 =2F\!\left[|\cos(2\theta)|\right],
 \label{eq:two-ZZ-curve}
\end{equation}
which gives Eq.~\eqref{eq:two-ZZ-saturation} at $\theta=\pi/4$.

\subsection{Saturation by pairwise $i$SWAP interactions}
\label{app:application-XY}

For either system qubit $Q\in\{A,B\}$ and its environmental partner, let
\begin{equation}
 H_{XY}=\sigma_Q^+\sigma_E^-+\sigma_Q^-\sigma_E^+.
\end{equation}
At $\theta=\pi/2$, $e^{-i\theta H_{XY}}$ is an $i$SWAP up to the convention for its phase.  Substitution in Eq.~\eqref{eq:application-global-state} gives
\begin{equation}
 H(T_1)=1,\qquad H(T_2)=2,\qquad H(T_1T_2)=1,
\end{equation}
and hence
\begin{equation}
 I(T_1:T_2)=2.
 \label{eq:single-iswap-MI}
\end{equation}
The two-environment process is the tensor product of two such processes. Relative-entropy Choi non-Markovianity is additive on this product, so Eq.~\eqref{eq:single-iswap-MI} gives
\begin{equation}
 \bm{\rmN}_{D}^{\mathrm{Choi}}\!\left[\Pcal_{\pi/2}^{(2E,XY)}\right]=2+2=4,
\end{equation}
proving Eq.~\eqref{eq:two-XY-saturation}.  For the common-environment model, the curve in Fig.~\ref{fig:memory-dimension}(b), together with the spatial and total curves in Fig.~\ref{fig:application-interpolation}(b), is evaluated directly from Eqs.~\eqref{eq:application-global-state}-- \eqref{eq:application-total-calculation}.

\subsection{Single-memory tightness with joint encoding}
\label{app:application-tightness}

The physical common-$ZZ$ model already saturates the classical $d_E=2$ bound, but its environmental label can be resolved from either system qubit at the saturation time.  A complementary construction stores the label exclusively in a joint degree of freedom of $AB$.  Consider
\begin{equation}
 U_{\rm cl}(\theta)=e^{-i\theta Z_AZ_BZ_E},
 \qquad \rho_E=\frac{\mathbb I_E}{2}.
 \label{eq:joint-classical-construction}
\end{equation}
At $\theta=\pi/4$, the two conditional $AB$ Choi vectors are orthogonal but have identical one-qubit reductions.  Therefore
\begin{equation}
 \bm{\rmN}_{D}^{\mathrm{Choi}}=\bm{\rmN}_{D}=1,
\end{equation}
while the classical bit is encoded in correlations between $A$ and $B$.

For the quantum bound, initialize $E$ in $\ketbra{0}{0}_E$, define $V_{AB}=\mathrm{CNOT}_{B\to A}$, and let
\begin{equation}
 U_{\rm q}=V_{AB}^{\dagger}\mathrm{SWAP}_{AE}V_{AB}.
 \label{eq:encoded-swap}
\end{equation}
This swaps $E$ with the logical qubit $L=A\oplus B$.  It can be compiled using only the star couplings $A$--$E$ and $B$--$E$ through the chronological CNOT sequence
\begin{equation}
 C_{E\to A},\quad C_{A\to E},\quad C_{B\to E},\quad C_{E\to A}.
\end{equation}
Applying Eq.~\eqref{eq:encoded-swap} in both temporal blocks transports a maximally entangled logical qubit across the temporal cut, and consequently
\begin{equation}
 \bm{\rmN}_{D}^{\mathrm{Choi}}=\bm{\rmN}_{D}=2.
\end{equation}
Equations~\eqref{eq:joint-classical-construction} and \eqref{eq:encoded-swap} show that both single-qubit-memory constants are tight even when the relevant information is associated with a joint degree of freedom of the two-qubit system.

\bibliography{bounds}

\end{document}